\documentclass[12pt]{iopart} 

\usepackage{iopams, amsthm, amssymb, amsfonts}

\pdfoutput=1
\usepackage{etex}
\reserveinserts{18}

\usepackage{cite}
\usepackage{latexsym}
\usepackage{rawfonts}
\usepackage{graphicx}
\usepackage[usenames,dvipsnames,svgnames]{xcolor}
\usepackage{bbm}
\usepackage{latexsym}
\usepackage{multirow}
\usepackage{rotating}
\usepackage{lscape}
\usepackage{graphicx} 
\usepackage{subfigure}
\usepackage{xcolor}
\usepackage{fancybox}
\usepackage{ifmtarg}

\input prepictex
\input pictex
\input postpictex

\usepackage{cmbright}
\usepackage[T1]{fontenc}
\usepackage{graphicx}

\def\q{\quad}

\def\2;{\;\;}

\def\eps{\epsilon}

\def\pr{{\prime}}

\def\NatN{{\mathbb N}}

\def\RealN{{\mathbb R}}

\def\ComplexN{{\mathbb C}}

\def\mathL{{\mathbb L}}

\def\shalf{{\sfrac{1}{2}}}

\def\Im{{\MF{I}m}\,}
\def\Re{{\MF{R}e}\,}

\newcommand\Arg{\hbox{\minefont$\MF{A}\MF{r}\MF{g}\,$}}

\def\o#1{\overline{#1}}

\def\Ref#1{(\ref{#1})}

\def\C#1{{\mathcal #1}}
\def\MF#1{{\mathfrak #1}}

\def\binom#1#2{{{#1}\choose{#2}}}
\def\Bi#1#2{{\binom{#1}{#2}}}

\def\Sfrac#1#2{\hbox{\large $\frac{#1}{#2}$}}
\def\sfrac#1#2{\hbox{\nor $\frac{#1}{#2}$}}

\def\LB{\left(}         \def\RB{\right)}
\def\LV{\left|}        \def\RV{\right|}
\def\LA{\left\langle}        \def\RA{\right\rangle}

\def\LC{\left\{}       \def\RC{\right\}}
\def\LH{\left[}        \def\RH{\right]}

\newtheorem{theorem}{Theorem}[section]
\newtheorem{corollary}{Corollary}[theorem]
\newtheorem{lemma}[theorem]{Lemma}
 
\def\nor{\normalsize}

\def\norf{\normalfont}

\def\svv{{\;\hbox{$|$}\;}}

\def\edge#1#2{{\langle #1 \hspace{0.85pt}{\sim}\hspace{0.85pt} #2 \rangle}}

\def\thin{ {\hspace{0.75pt}} }
\def\plus{{\hspace{0.5pt}{+}\hspace{0.5pt}}}
\def\minus{{\hspace{0.5pt}{-}\hspace{0.5pt}}}

\usepackage{cmbright}
\usepackage{textcomp}
\usepackage[font=small,labelfont=sf]{caption}
\newcommand*{\minefont}{\fontfamily{cmr}\selectfont}

\usepackage{subfig}

\definecolor{blue}{rgb}{0,0.18,0.39}
\definecolor{RoyalBlue}{rgb}{0,0.2,0.7}

\begin{document}
\title{Partition and generating function zeros in adsorbing self-avoiding walks}
\author{
E.J. Janse van Rensburg$^1$\footnote[1]{\texttt{rensburg@yorku.ca}}
}

\address{$^1$Department of Mathematics and Statistics, 
York University, Toronto, Ontario M3J~1P3, Canada\\}

\begin{abstract}
The Lee-Yang theory of adsorbing self-avoiding walks is presented.  It is shown that Lee-Yang
zeros of the generating function of this model asymptotically
accumulate uniformly on a circle in the complex plane, and that
Fisher zeros of the partition function distribute in the complex plane such that a 
positive fraction are located in annular regions centred at the origin.   These results
are examined in a numerical study of adsorbing self-avoiding walks in the square and 
cubic lattices.  The numerical data are consistent with the rigorous results; for example, Lee-Yang
zeros are found to accumulate on a circle in the complex plane and a positive fraction 
of partition function zeros appear to accumulate on a critical circle.  
The radial and angular distributions
of partition function zeros are also examined and it is found to be consistent with
the rigorous results.
\end{abstract}

\ams{82B41, 82B23}
\maketitle

\section{Introduction}
\label{section1}   

The partition function of a lattice spin model (such as the Ising or Potts model) is a 
sum over all spin configurations $\{\sigma_i\}$ in a lattice $\Lambda$ given by
\begin{equation}
Z_s(\beta,H) = \sum_{\{\sigma_i\}}
e^{-\beta\, {\cal H} (\{\sigma_i\})} .
\label{eqn1}   
\end{equation}
The inverse temperature is $\beta = \sfrac{1}{kT}$, and $\C{H}$ is the
\textit{Hamiltonian} of the model.  This is a function of the \textit{state} 
(configuration of lattice spins $\{\sigma_i\}$) and is given by
\begin{equation}
\C{H}(\{\sigma_i\}) = -J\sum_{\edge{i}{j}} f(\sigma_i , \sigma_j) - H \sum_{i} g(\sigma_i) 
+ BC.
\label{eqn2}   
\end{equation}
In this expression, $f(\sigma_i , \sigma_j)$ is a function of neighbouring spins
$\sigma_i$ and $\sigma_j$
connected by edges $\edge{i}{j}$.  In the Ising model the spins take
values $\sigma_i = \pm 1$ and $f(\sigma_i,\sigma_j) = \sigma_i\sigma_j$.  Similarly,
$g$ is a function of the spins $\sigma_i$ and in the Ising model $g(\sigma_i)
= \sigma_i$.  The Hamiltonian may also include terms corresponding to particular 
boundary conditions (denoted by $BC$). In Potts models, or other spin
models, the functions $f$ and $g$ may be defined differently.  

The parameter $J$ in equation \Ref{eqn2} is the coupling or interaction strength
between adjacent spins, and the parameter $H$ is an external (magnetic) field which 
couples to the spins.   Typically, the partition function is expressed in terms of 
activities $x = e^{\beta J}$ and $y=e^{\beta H}$. 
If $\Lambda$ is a finite lattice, then $Z_s(\beta,H)$ is a polynomial in $x$ and $y$ 
with non-negative coefficients and so it has no
zeros on the positive real axes in the $x$-plane or in the $y$-plane.

The free energy of the model is given by $F_s(\beta,H) = \log Z_s(\beta,H)$.
In finite lattices the singular points of $F_s(\beta,H)$ in the complex $\beta$- or $H$-planes
are found at the zeros of $Z_s(\beta,H)$.  $H$-plane zeros are called \textit{Lee-Yang zeros}  
\cite{YL52,LY52} and the $\beta$- or $T$-plane zeros are \textit{Fisher zeros} \cite{F64}.  
Lee-Yang or Fisher zeros of spin systems may have properties given by
the \textit{Lee-Yang theorem} \cite{YL52,LY52} which is based on the following theorem.

\begin{theorem}[Proposition 5.1.1 in Ruelle \cite{Ruelle83}]
Let $N=\{1,2,\ldots,n\}$, $S=\{i_1,i_2,\ldots,i_s\} \subseteq N$ and $S^\prime = N\setminus S$.  Let
$t^S = \prod_{j=1}^s t_{i_j}$ where $t_k\in\ComplexN$.  Suppose that $\{ A_{ij}\}_{i\not=j}$ 
is a family of real numbers such that $|A_{ij}| \leq 1$ 
for $i,j\in N$.  

Define the polynomial
\[ P_n(t_1,t_2,\ldots, t_n) = \sum_{S\subseteq N} t^S\;
\LB \prod_{i\in S} \prod_{j\in S^\pr} A_{ij} \RB \! . \]
Then, if $P_n(t_1,t_2,\ldots,t_n) = 0$ and $|t_j|\geq 1$ for $1\leq j \leq n-1$, it
follows that $|t_n|\leq 1$. \qed
\end{theorem}

A corollary of this theorem is the Lee-Yang theorem \cite{LY52}, namely

\begin{theorem}[The Lee-Yang theorem -- Theorem 5.1.2 in Ruelle \cite{Ruelle83}]
Let $N=\{1,2,\ldots,n\}$, $S=\{i_1,i_2,\ldots,i_s\} \subseteq N$ (so that $s=|S|$)
and $S^\prime = N\setminus S$.  Suppose that $\{ A_{ij}\}_{i\not=j}$ is a family
of real numbers such that $|A_{ij}| \leq 1$ 
for $i,j\in N$.

Define the polynomial
\[ P_n(t) = \sum_{S\subseteq N} t^{s}\;
\LB \prod_{i\in S} \prod_{j\in S^\pr} A_{ij} \RB \! . \]
Then the zeros of $P_n(t)$ all lie on the unit circle $|t|=1$ in the $t$-plane. \qed
\label{corr3.2}   
\end{theorem}

For example, if  $A_{ij}= \alpha$ for all $i,j$ in theorem \Ref{corr3.2},
$\alpha\in\RealN$ and $|\alpha|\leq 1$,  then the zeros of the polynomial
\[ P(t) = \sum_{S\in N} \alpha^{s(n-s)} t^{s}  = \sum_{s=0}^n \Bi{n}{s} \alpha^{s(n-s)} t^s . \]
all lie on the unit circle $|t|=1$ in the $t$-plane.

According to the Lee-Yang theorem, the zeros of the partition function of a wide class 
of lattice spin and related models (including the lattice gas and the Ising model) 
lie on the unit circle in the complex $y=e^{\beta H}$ plane.  In the thermodynamic
limit the Lee-Yang zeros accumulate on the unit circle, with a density profile which
determines the critical behaviour in the model.  The Lee-Yang zeros approach 
the positive real axis in the $y$-plane, creating an \textit{edge singularity} at the 
critical value of the magnetic field $H$.  This edge singularity is now known to be associated 
with a $\phi^3$ scalar field theory \cite{F78}, and has been studied in the
context of a wide variety of lattice models, including the $n$-vector model
\cite{KF79}, the $q$-state Potts model \cite{KC98,KC98a}, and also for a two dimensional
$\phi^3$ theory \cite{C85}.

In contrast to Lee-Yang zeros, the zeros in the complex temperature plane may not
be distributed on the unit circle (in the complex $x=e^{\beta J}$ plane).  More generally, 
they are distributed in ways which depend on the model, and are useful in understanding 
phase transitions \cite{F65}.  In the Ising model with special 
boundary conditions these zeros accumulate on two circles in the complex $x$-plane
\cite{F65,BK74}. These circles correspond to the ferromagnetic and anti-ferromagnetic
phases.  This pattern is not seen in the $q$-state Potts model  when $q>2$ \cite{KC98,KC01}.

The partition function zeros of models of lattice clusters have received far less 
attention in the literature.  Some results which are related to the Lee-Yang
singularities in spin systems are known.  For example, the Lee-Yang edge 
singularity of the Ising model in an imaginary magnetic field in $d$ dimensions
was shown to be related to critical behaviour in a model of lattice trees in 
$d\plus 2$ dimensions (see reference \cite{PS81}), giving exact values for the
lattice tree entropic exponent in two and in three dimensions.  Partition function
zeros of collapsing self-avoiding walks were studied in references \cite{L04,RSL14},
and considered for adsorbing walks in reference \cite{M78,TAP13,TLS14}.

In this paper the properties of Lee-Yang and Fisher zeros in a model of adsorbing
self-avoiding walks are considered.  Firstly, theorems about the distribution of
polynomial zeros  can be used to prove results on the distribution of generating
and partition function zeros.  These theorems include
classical theorems by Erd\"os and Tur\'an \cite{ET50}, and also newer results in
references \cite{HN08} and \cite{E08}.  
Secondly, by using the GAS algorithm \cite{JvRR09,JvR16} 
to estimate microcanonical data on adsorbing walks, a numerical approach 
to study the distribution of Lee-Yang zeros is presented.

The paper is organised as follows.

In section \ref{section2} a short review of models of adsorbing walks is given.
In particular, the partition function, generating function and free energy is
defined, and partial sums of the generating function is introduced.  

In sections \ref{Lee-Yang} and \ref{Yang-Lee}
the properties of zeros of partial sums of the generating function are
presented, while partition function zeros (Fisher zeros) are examined in section \ref{Fisher2}.

In general the partition function of interacting models of connected clusters 
(walks, polygons or animals) in a lattice is of the form
\begin{equation}
Z_n(a) = \sum_{m\geq 0} q_n (m)\, a^m
\end{equation}
where $q_n(m)$ is the number of clusters of size $n$ and energy $m$.  The zeros
$\{ a_k\}$ of $Z_n(a)$ are \textit{Fisher zeros} (or partition function zeros).  While
the results in section \ref{Fisher2} are directly applicable to the model of adsorbing
walks in this paper, it is also the case that similar approaches can be used in
other models. 

The partition function $Z_n(a)$ is a polynomial in 
$a$, and so it factors such that the (extensive) free energy of the model is given by
\begin{equation}
f_n(a) = \log Z_n(a) = \log C + \sum_{k\geq 0} \log (a\minus a_k)
\end{equation}
where $\log C$ is an analytic (background) contribution to $f_n(a)$.  At each zero $a_k$
there is a branch point in $f_n(a)$ in the complex $a$-plane, and, by taking
the thermodynamic limit $n\to\infty$, these singular points distribute
in the $a$-plane and create edge-singularities on the positive real $a$-axis.
These edge-singularities correspond to phase transitions in the model
\cite{LY52,YL52,Ruelle83}.  It is shown in this paper that, for adsorbing 
walks, a positive fraction of partition function zeros accumulate in an annulus with centre
at the origin (see theorem \ref{thm2.6}), and that the angular distribution
of the zeros is bounded (see theorem \ref{thm2.7a}) in the thermodynamic limit.

Related to the above are zeros of the generating function
\begin{equation}
G(a,t) = \sum_{n\geq 0} Z_n(a)\, t^n .
\end{equation}
These are studied in sections \ref{Lee-Yang} and \ref{Yang-Lee}
by considering partial sums of the generating
function, and examining the properties of $t$-plane zeros (these are
\textit{Lee-Yang zeros}) or the properties of $a$-plane zeros (these
are \textit{Yang-Lee zeros}).  Lee-Yang zeros are shown to accumulate
on a circle in the $t$-plane for models of adsorbing walks
(see theorems \ref{thm3.4} and \ref{thm3.5}), while
Yang-Lee zeros in the $a$-plane are found to have properties similar
to the partition function zeros examined in section \ref{Fisher2}.  

The numerical determination of the location of partition function zeros are
examined in section \ref{Numerical} and the properties of 
the square and cubic lattice adsorbing walk partition function 
and generating function zeros are considered numerically in 
sections \ref{Numerical2} and \ref{Numerical3}, respectively.  
The Lee-Yang zeros are found to distribute evenly on a
circle in the $t$-plane (approaching a uniform distribution with increasing
length of the walks), while partition function zeros are more widely
distributed.  It appears that a fraction of partition function zeros accumulate
on a circle of critical radius in the $a$-plane.  Similar results are found for
Yang-Lee zeros in the $a$-plane. In addition, the trajectories of the leading
partition function zeros were estimated, as well as the radial and angular
distribution of partition function and Yang-Lee zeros.

Some final comments are made in the conclusions in section \ref{conclusions}.

\begin{figure}[t]
\input Figures/figure01.tex
\caption{A positive self-avoiding walk in the half square lattice.  The walk
starts at the origin, and steps into the upper half plane on lattice sites.
Each vertex of the walk in the boundary of the half lattice is a visit (with 
the exception of the origin, which is not a visit).  This walk has 7 visits.}
\label{figure01}  
\end{figure}

\section{Lee-Yang and Fisher zeros in adsorbing self-avoiding walks}
\label{section2}

\subsection{Adsorbing self-avoiding walks}

Let $\mathL^d$ be the $d$-dimensional hypercubic lattice, and define the half-lattice
$\mathL_+^d =\{ \edge{\vec{u}}{\vec{v}}\in \mathL^d 
\svv \hbox{$u_d\geq0$ and $v_d\geq 0$} \}$ where $\edge{\vec{u}}{\vec{v}}$ 
is a lattice edge between adjacent vertices $\vec{u}$ and $\vec{v}$, and 
$u_d$ is the $d$-th component of $\vec{u}$.  The boundary of $\mathL_+^d$ is
denoted by $\partial \mathL_+^d$ and it is a $(d\minus 1)$-dimensional hypercubic lattice
containing the origin $\vec{0}$. 

A self-avoiding walk from $\vec{0}$ in $\mathL^d_+$ is a \textit{positive walk}.
Denote the number of positive walks of length $n$ with $v$ vertices (called 
\textit{visits}, and excluding $\vec{0}$) in $\partial \mathL^d_+$ by $c_n^+(v)$.
See figure \ref{figure01}.

The total number of positive walks from the origin in $\mathL^d_+$
is $c_n^+ = \sum_v c_n^+(v)$.  The expected asymptotic behaviour of $c_n^+$ is
given by
\begin{equation}
c_n^+ \simeq B\,n^{\gamma_1-1}\,\mu_d^n
\label{eqn3}   
\end{equation}
where $\mu_d$ is the growth constant of the self-avoiding walk \cite{H57}.  The
entropic exponent of half-space walks $\gamma_1$ has exact value
$\sfrac{61}{64}$ in two dimensions \cite{C83}, and approximate value
$0.679(2)$ in three dimensions \cite{HG94} (this estimate was improved to
$0.677667(17)$ in reference \cite{CCG15}).  The growth constant has values
$\mu_2 = 2.638\ldots$ \cite{CJ12} and $\mu_3=4.684\ldots$ \cite{C13}.

The partition function of a model of adsorbing self-avoiding walks is given by
\begin{equation}
A^+_n(a) = \sum_{v=0}^n c_n^+ (v)\, a^v 
\label{eqn4}   
\end{equation}
where $a = e^{\beta J}$ is an activity conjugate to the number of visits
of the positive walks in $\partial\mathL^d_+$.  The \textit{finite size free energy}
of the model is $F_n(a) = \log A^+_n(a)$.

The thermodynamic limit of this model is given by the limit
\begin{equation}
\C{A}^+(a) = \lim_{n\to\infty}\sfrac{1}{n}\, F_n(a)  = \lim_{n\to\infty} \sfrac{1}{n}\, \log A_n^+(a) .
\label{eqn5}   
\end{equation}
This was shown in reference \cite{HTW82}, and see section 9.1 in reference \cite{JvR15}
for more results on this model.   $\C{A}^+(a)$ is a convex function of $\log a$.
Of particular interest in this model is the existence of a critical value $a_c^+$ 
of the activity $a$ where the walk adsorbs onto $\partial\mathL_+^d$.  It is known 
that $\C{A}^+(a)$ is not analytic at $a_c^+$, that
$a_c^+ < \sfrac{\mu_d}{\mu_{d-1}}$ (see, for example, corollary 9.3 in
reference \cite{JvR15}) and that $a^+_c>1$ (see theorem 9.10 in 
reference \cite{JvR15}, or reference \cite{JvR98} for adsorbing lattice 
polygons).  In $\mathL^2_+$ series enumeration gives 
$a_c^+ = 1.77564\ldots$ \cite{BGJ12} and Monte Carlo simulations in 
$\mathL^3_+$ give $a_c^+ = 1.306(7)$ \cite{JvR16}.

It is known that \cite{HTW82}
\begin{equation}
\C{A}^+(a)  \; \cases{
= \log \mu_d, & \hbox{if $a\leq a_c^+$}; \\
> \log \mu_d, & \hbox{if $a>a_c^+$}. 
}
\label{eqn6}   
\end{equation}The \textit{energy}
(or density of visits) $\C{V}(a) = \lim_{n\to\infty}\sfrac{1}{n} \LA v_n \RA$ (where
$\LA v_n \RA$ is the expected number of visits for walks of length $n$)
is the order parameter of the model. By convexity it follows that 
$\C{V}(a) = a\sfrac{d}{da} \C{A}^+(a) =0$ if $0<a<a_c^+$ 
(this is the \textit{desorbed phase}), and $\C{V}(a) = a\sfrac{d}{da}\C{A}^+(a) >0$ 
if $a>a_c^+$ (whenever $\C{A}^+(a)$ is differentiable -- this is the
\textit{adsorbed phase}).   Notice that 
$\C{A}^+(a)$ is differentiable for almost all $a  >  0$.  

In the desorbed phase the walk tends to make few visits in 
the adsorbing plane $\partial \mathL^d_+$ (on average, $o(n)$ for 
walks of length $n$).  That is, in the desorbed phase the walk tends to drift 
away from $\partial \mathL^d_+$ and explores the bulk
of the lattice $\mathL^d_+$. In the adsorbed phase $\C{V}(a)>0$,
and the walk returns to $\partial\mathL^d_+$ with positive frequency. These
two phases are separated by the adsorption critical point $a_c^+$.

The asymptotic behaviour of $A_n^+(a)$ are given by expressions similar to equation
\Ref{eqn3}.  In particular, it is expected that
\begin{equation}
A_n^+(a) \sim
\cases{
B_-\,n^{\gamma_1-1}\, \mu_d^n, & \hbox{\norf desorbed phase ($a<a_c^+$)}; \\
B_0\,n^{\gamma_s-1}\, \mu_d^n, & \hbox{\norf critical point ($a=a_c^+$)}; \\
B_+\,n^{\gamma^{(d-1)}-1}\, e^{n\,\C{A}^+(a)}, 
& \hbox{\norf adsorbed phase ($a>a_c^+$)} .
}
\label{eqn7}   
\end{equation}
The exponent $\gamma_s$ is a \textit{surface
exponent} and has exact value $\gamma_s = \sfrac{93}{64}$ in two
dimensions \cite{BEG89}.  In three dimensions, it was estimated that
$\gamma_s = 1.304(16)$ \cite{LM88A,ML88A}.  In the adsorbed phase the exponent
$\gamma^{(d-1)}$ is the entropic of exponent of self-avoiding walks in one
dimension lower.

The scaling of $\C{A}^+(a)$ is controlled by the crossover exponent $\phi$
which has exact value $\shalf$ in two dimensions \cite{BY95,BEG89}, 
and mean field value $\shalf$ \cite{BEG89}.
It is thought to have value $\shalf$ in all dimensions $d\geq 2$ and in $d=3$
was estimated to be $\phi=0.505(6)$ \cite{JvR16}.  Since there
is a bulk entropy contribution of $\log \mu_d$ in $\mathL^d_+$, the free energy
scales as
\begin{equation}
\C{A}^+(a) \sim \log \mu_d + \sfrac{1}{n}\, f_a (n^\phi\,(a\minus a_c^+)),
\end{equation}
where $f_a$ is a scaling function such that $f_a(x) = 0$ if $a<a_c^+$.

The \textit{grand partition} or \textit{generating function} 
of adsorbing walks is defined by
\begin{equation}
G(a,t) = \sum_{n=0}^\infty A_n^+(a)\, t^n = 
\sum_{n=0}^\infty \sum_{v=0}^n c_n(v)\, a^v \; t^n .
\label{eqn9}   
\end{equation}
The circle of convergence of $G(a,t)$, in the $t$-plane, is evidently
given by $|t|=e^{-\C{A}^+(a)}$, since by the root test,   $G(a,t)$ is convergent if
\begin{equation}
|t| < t_c = \lim_{n\to\infty} | A_n^+(a) |^{-1/n} .
\end{equation}
There is a singular point in $G(a,t)$ on the positive real axis at $t_c$, and, if
$a>0$, then $t_c = e^{-\C{A}^+(a)}$ by equation \Ref{eqn5}.  In the desorbed phase
$a\in [0, a^+_c)$ and $t_c = \mu_d^{-1}$, by equation \Ref{eqn6}.

The partial sums of the generating function is defined by
\begin{equation}
G_N(a,t) = \sum_{n=0}^N A_n^+(a)\, t^n = 
\sum_{n=0}^N \sum_{v=0}^n c_n(v)\, a^v \; t^n 
\label{eqn11}   
\end{equation}
so that $G_N \to G$ as $N\to\infty$ and $|t| < t_c$.  For fixed $N$ and $a$, $G_N(a,t)$ is
a polynomial with non-negative coefficients, and is non-decreasing along the 
positive real axis.   The grand potential of walks of length at most $N$ is
$\log G_N(a,t)$.  This is singular at the zeros of $G_N(a,t)$.  Recovering $\log G(a,t)$ by
taking $N\to\infty$ shows that a singular point develops on the real axis at $t_c$.  This
is an \textit{edge-singularity}, which is the result of $t$-plane singular points in 
$G_N(a,t)$ converging to the real axis at $t_c$ as $N\to\infty$.

\subsection{Lee-Yang zeros}
\label{Lee-Yang}

Theorem \ref{corr3.2} is very useful in analysing lattice spin systems, such as the 
lattice gas \cite{Ruelle83}.  However, it seems more difficult to apply this result
to a model of adsorbing self-avoiding walks.  
Instead, one could use the theorem by Hughes and Nikechbali \cite{HN08}, which
will be stated below. 

Let $P(t)$ be a polynomial with zeros $\LA t_1,t_2,\ldots,t_N\RA$.
Follow reference \cite{HN08} and proceed by defining two functions $\nu_N$ and 
$\alpha_N$ to measure the distribution of zeros: 
\begin{eqnarray}
\nu_N(\rho) &= \#
\{ t_k \vert\, (1\minus \rho)\leq |t_k| \leq \sfrac{1}{1-\rho} \} ; \cr
\alpha_N(\theta,\phi) &= \#
\{ t_k \vert\, \theta< \Arg t_k  \leq \phi \} . 
\label{eqn12}   
\end{eqnarray}
Here $\Arg t_k$ is the principal argument of $t_k$ in the $t$-plane.
Then $\nu_N(\rho)$ is the number of zeros in a thin annulus about the unit circle (and of width
$2\rho + O(\rho^2)$).  If the zeros cluster about the unit circle, then $\nu_N (\rho) \simeq N$
as $N\to\infty$, for small $\rho$.  If $\lim_{N\to\infty} \sfrac{1}{N}\, \nu_N(\rho) = 1$ for
any $\rho\in(0,1)$, then almost all the zeros accumulate on the unit circle in 
the $t$-plane as $N\to\infty$.

The function $\alpha_N(\theta,\phi)$ is a measure of the
angular distribution of zeros.  If the distribution is uniform in the limit that
$N\to\infty$, then $\lim_{N\to\infty} \sfrac{1}{N}\, \alpha_N(\theta,\phi) = \sfrac{1}{2\pi}\,
(\phi\minus\theta)$ for $-\pi \leq \theta < \phi \leq \pi$.

\begin{theorem}[Hughes and Nikechbali; see theorem 3 in \cite{HN08}]
Let $\LA a_k \RA$ be a sequence of complex numbers such that $a_0a_N\not=0$.
Let $P(t) = \sum_{n=0}^N a_n\,t^n$.  Define
\[ L_N = \log \sum_{n=0}^N |a_n| - \sfrac{1}{2} \log |a_0| - \sfrac{1}{2} \log |a_N| . \]
Then, for $0< \rho < 1$,
\[ 1\minus \sfrac{1}{N}\, \nu_N(\rho) \leq \sfrac{2}{N\rho}\, L_N . \]
In, particular, if $L_N= o(N)$, then $\lim_{N\to\infty} \sfrac{1}{N}\, \nu_N(\rho) = 1$ for
any $\rho\in (0,1)$. \qed
\label{thm3.3}   
\end{theorem}

Theorem \ref{thm3.3} may be applied to the model of adsorbing walks.
Assume that the asymptotic behaviour of $A_n^+(a)$ is as given in equation \Ref{eqn7}.  Define
\begin{equation}
\lambda_a = \cases{
\mu_d, & \hbox{if $a\leq a_c^+$}; \cr
e^{\C{A}^+(a)}, & \hbox{if $a> a_c^+$}.
}
\label{eqn13}   
\end{equation}

By equation \Ref{eqn5}, 
\begin{equation}
\lim_{n\to\infty} \LH \Sfrac{A_n^+(a)}{\lambda_a^n} \RH^{1/n} = 1,\q
\hbox{and $A_n^+(a) = \lambda_a^{n+o(n)}$}.
\label{eqn14}   
\end{equation}
Replace $t$ by $\lambda_a^{-1}t$ in equation \Ref{eqn11} to define the partial sum
\begin{equation}
g_N(a,t) = \sum_{n=0}^N \LB\Sfrac{A_n^+(a)}{\lambda_a^n}\RB \! t^n .
\label{eqn15}   
\end{equation}
Taking $N\to\infty$ gives a generating function with circle of convergence $|t|=1$ for
any finite $a>0$.

Put $a_n = \sfrac{A_n^+(a)}{\lambda_a^n}$ in theorem \ref{thm3.3}.  Then the function $L_N$ 
becomes
\[ L_N = \log \sum_{n=0}^N \LV\Sfrac{A_n^+(a)}{\lambda_a^n} \RV 
- \sfrac{1}{2} \log 1 - \sfrac{1}{2} \log \LV\Sfrac{A_N^+(a)}{\lambda_a^N} \RV . \]
Since $A_n^+(a) = \lambda_a^{n+o(n)}$, this becomes
\[ L_N \leq \log \LB (N\plus 1)\, \lambda_a^{m_N} \RB + \sfrac{1}{2} \log \lambda_a^{o(N)} \]
where $m_N$ is that value of $o(n)$ which maximizes $\LV\sfrac{A_n^+(a)}{\lambda_a^n} \RV
= \LV \lambda_a^{o(n)}\RV$ (and so $m_N = o(N)$).  Divide by $N$ and take $N\to\infty$ to see that
\begin{equation}
 \lim_{N\to\infty} \sfrac{1}{N} \; L_N  = 0
\label{eqn16}   
\end{equation}
and so $L_N = o(N)$.  It follows by theorem \ref{thm3.3} that
$\lim_{N\to\infty} \sfrac{1}{N}\, \nu_N(\rho) = 1$ for any $\rho\in(0,1)$, 
and so the $t$-plane zeros of the partial generating function $g_N(a,t)$ 
cluster on the unit circle as $N\to\infty$. The result is the following
corollary of theorem \ref{thm3.3}:

\begin{theorem}
The zeros of $g_N(a,t)$ converge on the unit circle $|t| = 1$ in the $t$-plane as
$N\to\infty$ in the sense that, for any $\rho\in (0,1)$ and $\epsilon>0$, there is an $N_0$ 
such that for all $N\geq N_0$, $|1 - \sfrac{1}{N}\,\nu_N(\rho)| < \eps$ .  \qed
\label{thm3.4}  
\end{theorem}

This implies that the zeros of $G_N(a,t)$ accumulate on the circle 
$|t|=\lambda_a^{-1}$ as $N\to\infty$.   Hence, if
\begin{equation}
\psi_N(\rho) = \#
\{ t_k \vert\, (1\minus \rho) \lambda_a^{-1} \leq |t_k| 
\leq \sfrac{1}{1-\rho}\lambda_a^{-1} \} ,
\end{equation}
then a corollary of theorem \ref{thm3.4} is

\begin{corollary}
The zeros of $G_N(a,t)$ converge on the circle $|t| = \lambda_a^{-1}$ in the $t$-plane as
$N\to\infty$ in the sense that, for any $\rho\in (0,1)$ and $\epsilon>0$, there is an $N_0$ 
such that for all $N\geq N_0$, $|1 - \sfrac{1}{N}\,\psi_N(\rho)| < \eps$ .  \qed
\label{thm3.4C}  
\end{corollary}

That is, $\lim_{N\to\infty} \sfrac{1}{N}\, \psi_N(\rho) = 1$ for any $\rho\in(0,1)$.

The angular distribution of zeros can be considered using a theorem of Erd\"os and Turan
which is similar to theorem \ref{thm3.3} above.

\begin{theorem}[Erd\"os and Tur\'an  \cite{ET50}]
Suppose $\LA a_n\RA$ is a sequence in $\ComplexN$ and suppose that $a_0a_N \not=0$.
Define $P(t) = \sum_{n=0}^n a_n \, t^n$.  Define
\[ L_N = \log \sum_{n=0}^N |a_n| - \shalf \log |a_0| - \shalf \log |a_N| . \]
Then
\[ \LV \sfrac{1}{N}\, \alpha_N(\theta,\phi) - \sfrac{1}{2\pi}\,(\phi \minus\theta) \RV^2
\leq \Sfrac{C}{N}\, L_N ,\]
for some constant $C$. \qed
\label{thm3.6}  
\end{theorem}

A corollary of the Erd\"os and Turan theorem for adsorbing walks is a consequence 
of equation \Ref{eqn16}.  Since $\sfrac{1}{N}\, L_N\to 0$ as $N\to\infty$, the result is that
$\sfrac{1}{N}\, \alpha_N(\theta,\phi) \to \sfrac{1}{2\pi}\,(\phi\minus\theta)$.  That
is, the principal arguments of $t$-plane zeros are asymptotically uniformly distributed 
over angles about the origin.

\begin{theorem}
The arguments of the zeros of $G_N(a,t)$ approach a uniform distribution
over angles about the origin in the sense that, 
for any $-\pi \leq \theta < \phi \leq \pi$ and $\eps>0$, there is
an $N_1$ such that for all $N\geq N_1$, $|\sfrac{1}{N} \, \alpha_N(\theta,\phi)
- \sfrac{1}{2\pi}\,(\phi\minus\theta)|^2 \leq \eps$.  \qed
\label{thm3.5}  
\end{theorem}

Thus, by theorems \ref{thm3.4} and \ref{thm3.5} the $t$-plane zeros of the
partial generating function $G_N(a,t)$ are uniformly distributed on the circle 
of radius $\lambda_a^{-1}$ in the limit as $N\to\infty$.  In the thermodynamic limit this 
circle is a natural boundary of $G(a,t)$ in the $t$-plane, and
the point $\lambda_a^{-1}$ on the positive real axis is an edge singularity.

\subsection{Fisher zeros}
\label{Fisher2}

Let $\LA a_j \RA$ be the set of zeros of $A_n^+(a)$ in the complex 
$a$-plane.  Since $A_n^+(a)$ is a polynomial of degree $n$ in $a$, there
are $n$ such zeros, counted with multiplicity.  The \textit{leading zero}
will always be denoted by $a_1$, and it is defined as that zero with smallest
positive principal argument.

By the fundamental theorem of algebra $A_n^+(a)$ factors into a product
over linear factors of the form
\begin{equation}
A_n^+(a) = C\prod_{j=1}^n (a\minus a_j).
\label{eqn18}   
\end{equation}
The coefficients in (the polynomial) $A_n^+(a)$ are all positive integers, so
the zeros $a_j$ are all in pairs of complex conjugates or are real, and, since
$A_n^+(a)$ is a polynomial with non-negative coefficients, none of the $a_j$ are
located on the positive real axis in the complex $a$-plane.

The coefficient of $a^n$ is equal to $c_n^+(n) = c_n^{(d-1)}$ (the number 
of walks from the origin in $\partial \mathL^d_+ \equiv \mathL^{(d-1)}$, and this
fixes the constant $C=c_n^{(d-1)}$.  
The (finite size extensive) free energy $F_n(a)$ is given by 
\begin{equation}
F_n(a) = \log A_n^+(a) = \log c_n^{(d-1)} + \sum_{j=1}^n \log (a\minus a_j) .
\label{eqn19f}   
\end{equation}
$F_n(a)$ is a non-analytic function of $a$ but is analytic on points along the positive
real axis.  Dividing by $n$ and taking $n\to\infty$ shows that the limiting free energy
is given by
\begin{equation}
\C{F}(a) = \lim_{n\to\infty} \sfrac{1}{n}\, F_n(a)
= \log\mu_{d-1} + \varrho (a),
\label{eqn19q}   
\end{equation}
where the limit
\begin{equation}
\varrho (a) = \lim_{n\to\infty} \sfrac{1}{n}\, \sum_{j=1}^n \log (a\minus a_j)
\end{equation}
exists and 
is the limiting average of $\log(a\minus a_j)$ taken over the zeros
of $A_n^+(a)$.  This limit exists for $a>0$ since $\C{F}(a)$ is 
known to exist in this model \cite{HTW82}.  Since $\varrho(a)$ is analytic
on a line segment in the complex $a$-plane, its analytic continuation into 
the complex $a$-plane is unique up to natural boundaries, by the coincidence 
principle.

\subsubsection{The distribution of Fisher zeros:}
\label{Dist-Fisher}   
Define the plane measure of a complex set $E\subseteq \ComplexN$ 
by noting that $\ComplexN \simeq \RealN^2$ and by considering $E$ 
to be a subset of $\RealN^2$ by inclusion.  This endows $\ComplexN$
with (real valued) plane measure $\alpha$ and a $\sigma$-algebra of
$\alpha$-measure sets.

The distribution of Fisher zeros in the $a$-plane will be analysed by constructing
a probability measure on $\ComplexN$ which is absolutely continuous with respect to
plane measure $\alpha$.
 
Suppose that $\langle h_n (a) \rangle$ is a family of (non-negative, real-valued
and normalised) distribution functions on the complex $a$-plane with respect 
to plane measure $\alpha$, and with the property that
\begin{equation}
\lim_{n\to\infty} h_n(a) = 0,\q\hbox{if $|a|>0$}. 
\end{equation} Then a distribution function
on the $j$-th zero $a_j$ of $A_n^+(a)$ is given by $h_{j,n}(a) = h_n(a-a_j)$.  

A distribution function $H_n(a)$ on the set of zeros $\{ a_j \}$ may be defined by
\begin{equation}
H_n(a) = \sfrac{1}{n} \sum_{j=1}^n h_{j,n}(a) .
\label{eqn26Z}  
\end{equation}
$H_n(a)$ may be used to define a set function $\lambda_n$ on plane measurable 
sets $E$ in $\ComplexN$:
\begin{equation}
\lambda_n\,E = \int_E H_n(a)\,d\alpha(a) 
= \sfrac{1}{n} \int_E \sum_{j=1}^n h_{j,n}(a) \, d\alpha(a) .
\end{equation}
Notice that $\lambda_n\, E = 0$ if $\alpha E = 0$, that $\lambda_n E \leq \lambda_n F$
if $E\subseteq F$ and $\lambda_n E\cup F  = \lambda_n E +\lambda_n F$ if $E\cap F
=\emptyset$.  This shows that $\lambda_n$ is a complete measure on $\alpha$-measure sets
in $\ComplexN$.  Since $\lambda_n \ComplexN = 1$, it is also a probability measure
on $\alpha$-measure sets in $\ComplexN$ and it is absolutely continuous with respect
to plane measure $\alpha$.

There is a great deal of freedom in constructing the measure $\lambda_n$.  In this
paper two examples of the measure $\lambda_n$ are considered:
\begin{itemize}
\item Define the family of functions 
\begin{equation}
 f_n(a) = \cases{
\Sfrac{\tau_n}{\pi}, & \hbox{\norf if $|a|\leq \sfrac{1}{\sqrt{\tau_n}}$}\, ;\\
0, & \hbox{\norf otherwise,}
} \label{eqn28BB}
\end{equation}
where $\tau_n>0$ is a non-decreasing function such that $\tau_n\to\infty$ as $n\to\infty$.
Then $f_n(a) \to 0$ if $|a|>0$ and $n\to\infty$.  Put $f_{j,n}(a) = f_n(a\minus a_j)$
and define the measure $\mu_n$ by 
\begin{equation}
\mu_n E = \sfrac{1}{n} \int_E \sum_{j=1}^n f_{j,n} (a) \, d\alpha (a) 
\label{eqn10q}   
\end{equation}
on the $\sigma$-algebra of $\alpha$-measurable sets $E$.
\item Similarly, define the family of Gaussians on $\ComplexN$ by
\begin{equation}
 g_n (a) = \tau_n \,
e^{-\pi\tau_n\,|a|^2} , 
\label{eqn28CC}
\end{equation}
where $\tau_n>0$ is a non-decreasing function such that $\tau_n\to\infty$ as $n\to\infty$.
Notice that $g_n(a) \to 0$ if $|a|>0$ and $n\to\infty$.  Define
$g_{j,n}(a) = g_n(a\minus a_j)$ and define the measure $\nu_n$ by
\begin{equation}
\nu_n E = \sfrac{1}{n} \int_E \sum_{j=1}^n g_{j,n} (a) \, d\alpha(a) 
\label{eqn28z}
\end{equation}
on the $\sigma$-algebra of $\alpha$-measurable sets $E$.
\end{itemize}
Then both $\mu_n \ll \alpha$ and $\nu_n \ll \alpha$ and one may choose
$\lambda_n = \mu_n$, or $\lambda_n=\nu_n$ in what follows.

Existence of the limit
\begin{equation}
\lambda E = \lim_{n\to\infty} \lambda_n E
\label{eqn10a}   
\end{equation}
for arbitrary $\alpha$-measure sets $E$ is dependent on the model and the 
choice of $\lambda_n$, and may be difficult to prove.  We assume
that $\lambda$ is defined in this way, and that this limit exists, and will have
this as a condition in the theorems below.

If $\lambda$ exists, then it is monotone.  This follows because, if
$E\subseteq F$, then $\lambda_n E \leq \lambda_n F$ for all values $n$.
This shows that $\lambda E \leq \lambda F$ if $E\subseteq F$.   In addition, 
if $E\cap F= \emptyset$, then $\lambda_n E\cup F = \lambda_n E \plus \lambda_n F$
for all values of $n$, and so it follows that 
 $\lambda E\cup F = \lambda E \plus \lambda F$.  Thus, $\lambda$ is
additive over the unions of disjoint sets. 
In addition, $\lambda \emptyset = 0$ and $\lambda \ComplexN = 1$.
Hence, $\lambda$ is a complete probability measure on $\ComplexN$ (on
the $\sigma$-algebra of $\alpha$-measure sets).

\begin{lemma} 
If the limit $\lambda E = \lim_{n\to\infty} \lambda_n E$ exists for all $\alpha$-measure
sets $E$, then the set function $\lambda$ is a complete probability measure 
on the $\sigma$-algebra of $\alpha$-measure sets in $\ComplexN$. \qed
\label{lemma21}   
\end{lemma}

Since $\lambda_n$ is absolutely continuous with respect to plane measure
$\alpha$ on $\ComplexN$, it follows (by the Radon-Nikodym theorem) that
there exist non-negative distribution functions $p$ and $p_n$, unique up
to zero $\alpha$-measure sets, such that for 
any $\alpha$-measure set $E$ in $\ComplexN$,
\begin{equation}
\lambda E = \int_E p \, d\alpha ,
\q\hbox{and}\q
\lambda_n E = \int_E p_n \, d\alpha.
\end{equation}

It is next shown that $p_n \to p$ in measure.  Define set 
functions $\beta_n$ on $\alpha$-measure sets $E$ by
\begin{equation}
\beta_n  E = ( \lambda_n \minus \lambda ) E = \int_E (p_n\minus p)\,d\alpha .
\end{equation}
The $\beta_n$ are signed measures on $\ComplexN$.  Assuming that $\lambda$
is defined by equation \Ref{eqn10a} it follows that
\begin{equation}
\lim_{n\to\infty} \beta_n E = \lim_{n\to\infty} (\lambda_n\minus \lambda) E = 0
\end{equation}
for any $\alpha$-measure set $E$.

The Hahn decomposition of $\ComplexN \simeq \RealN^2$ by the measure $\beta_n$
decomposes $\ComplexN$ in a positive set $F_+^{(n)}$ and a negative
set $F_-^{(n)}$. Define the mutually singular (positive) 
measures $\beta_n^+$ and $\beta_n^-$ on the Hahn decomposition 
of $\ComplexN$:
\begin{equation}
\beta_n^+E = \beta_n(E\cap F_+^{(n)}),
\q\hbox{and}\q
\beta_n^-E = -\beta_n(E\cap F_-^{(n)}) .
\end{equation}
Then it follows that $\beta_n = \beta_n^+ - \beta_n^-$ is a Jordan decomposition
of the signed measure $\beta_n$.
Since $\lim_{n\to\infty} \beta_n E = 0$ for any $\alpha$-measure set $E$, it follows 
that $\lim_{n\to\infty} \beta_n^+ E = 0$ and $\lim_{n\to\infty} \beta_n^- E = 0$
for any $\alpha$-measure set $E$.

The total variation of the signed measure $\beta_n$ is given by
\begin{equation}
|\beta_n|  = \beta_n^+  + \beta_n^- 
\end{equation}
and it is a measure on the $\sigma$-algebra of $\alpha$-measure sets
in $\ComplexN$.  It follows that $\lim_{n\to\infty} |\beta_n| E = 0$ 
for any $\alpha$-measure set $E$.  Notice that
\begin{eqnarray}
|\beta_n| E 
&=&  \beta_n(E\cap F_+^{(n)})  -\beta_n(E\cap F_-^{(n)}) \cr
&=& \int_{E\cap F_+^{(n)}} (p_n\minus p)\,d\alpha
 + \int_{E\cap F_-^{(n)}} (p\minus p_n)\,d\alpha \cr
&=& \int_E |p_n \minus p| \, d\alpha .
\end{eqnarray}
Since $\lim_{n\to\infty} |\beta_n| E = 0$ for every $\alpha$-measure set $E$
in $\ComplexN$, it follows that 
\begin{equation}
\lim_{n\to\infty}\int_E |p_n-p|\,d\alpha = 0.
\label{eqn22A}   
\end{equation}
Choosing $E=\ComplexN$ shows that
$\| p_n \minus p \|_1 \to 0$ in the normed space $L^1(\alpha)$.  
That is, if the distribution functions $p_n$ and $p$ exist, then $p_n$ converges to
$p$ in the $L^1(\alpha)$ norm.  This gives the lemma:

\begin{lemma}
Suppose that the limit $\lambda E = \lim_{n\to\infty} \lambda_n E$ exists.
Then there exists distribution functions $p$ and $p_n$ such that
\[ 
\lambda E = \int_E p \, d\alpha ,
\q\hbox{and}\q
\lambda_n E = \int_E p_n \, d\alpha.
\]
Moreover, for every $\alpha$-measure set $E$, $\lim_{n\to\infty}\| p_n \minus p \|_1 = 0$.
This implies that 
\[ \lim_{n\to\infty} p_n = p,\q\hbox{in measure} . \]
\end{lemma}

\begin{proof}
To see this, notice that by equation \Ref{eqn22A}
$p_n \to p$ in the normed space $L^1(\alpha)$.
Let $\eps>0$ and define measure sets 
$F_n = \{ z\in\ComplexN \svv |p_n(z) \minus p(z)|> \eps \}$.  Then
\begin{equation*}
\int |p_n \minus p| \, d\alpha \geq \int_{F_n} |p_n\minus p|\, d\alpha
\geq \eps\,\alpha F_n .
\end{equation*}
This shows that $\lim_{n\to\infty} \alpha F_n = 0$ and so
$\lim_{n\to\infty} p_n= p$ in measure.
\end{proof}

By the Lebesgue Dominated Convergence theorem for sequences
which converge in measure it follows that for any $\alpha$-measure
set $E$, since $p_n \to p$ in measure,
\begin{eqnarray}
\lambda E = \lim_{n\to\infty} \lambda_n E  
= \lim_{n\to\infty} \int_E p_n \, d\alpha 
= \int_E p\, d\alpha.
\end{eqnarray}
Since $\lambda$ and $\lambda_n$ are finite measures, and $\alpha$ is
$\sigma$-finite, the Lebesgue Dominated Convergence theorem
for sequences which converge in measure may be applied to 
the function $\varrho(a)$ in equation \Ref{eqn19q}.  That is, for 
any $a>0$ on the positive real axis,
\begin{eqnarray}
\varrho (a) 
&=& \lim_{n\to\infty} \sfrac{1}{n} \sum_{j=1}^n \log(a\minus a_j) \cr
&=& \lim_{n\to\infty} \int_\ComplexN \log (a\minus z)\, d\lambda_n(z) \cr
&=& \lim_{n\to\infty} \int_\ComplexN \log (a\minus z)\, p_n(z)\, d\alpha(z) \cr
&=& \int_\ComplexN \log (a\minus z)\, p(z)\, d\alpha(z)  = \int_\ComplexN \log (a\minus z) \, d\lambda (z) .
\end{eqnarray}
In other words, in terms of plane measure $\alpha$ and the limiting
probability distribution function $p(z)$,
\begin{equation}
\varrho (a) = \int_E \log (a\minus z)\, p(z)\, d\alpha (z) ,
\label{eqn11c}   
\end{equation}
where the integration is over the entire complex plane with real plane measure
$\alpha$.  Notice that for any $\alpha$-measure set $E$, it similarly follows that
\begin{equation}
 \lim_{n\to\infty} \sfrac{1}{n} \sum_{a_j\in E} \log(a\minus a_j) 
= \int_E \log (a\minus z)\, p(z)\, d\alpha(z) .
\end{equation}
This shows that the distribution $p(z)$ is unique up to zero $\alpha$-measure sets.

Since $C = c_n^{(d-1)}$, the following theorem follows from equations \Ref{eqn19q} 
and \Ref{eqn11c}.

\begin{theorem}
Suppose that the partition function $A_n^+(a)$ of adsorbing walks of length $n$ 
has $a$-plane zeros denoted by $\langle a_j \rangle$.

In addition, suppose that $h_n(a)$ is a family of distribution functions on $\ComplexN$ 
such that $h_n(a) \to 0$ if $|a|>0$ and $n\to\infty$ and suppose that $\alpha$ is plane
measure on $\ComplexN$.  Then the set-function
\[
\lambda_n\,E = \sfrac{1}{n} \int_E \sum_{j=1}^n h_{j,n}(a) \, d\alpha(a) 
\]
where $h_{j,n} (a) = h_n(a\minus a_j)$  is a probability measure on the $\sigma$-algebra of
plane measure sets in $\ComplexN$.

Moreover, if the limit $\lambda E = \lim_{n\to\infty} \lambda_n E$ exists,
then there exists a distribution function $p(z)$ on $\ComplexN$, unique up
to zero $\alpha$-measure sets, such that
the limiting free energy of adsorbing walks is given by
\[
\C{A}^+(a) = \log \mu_{d-1} + \varrho (a) = \log\mu_{d-1} 
+ \int_\ComplexN \log (a\minus z)\, p(z)\,d\alpha(z) ,
\]
where $a\in\RealN^+$ (the positive real axis). \qed
\label{eqn12cc}   
\end{theorem}

The function $\C{A}^+(a)$ can be analytically continued up to natural
boundaries in the $a$-plane.
Since $\C{A}^+(a) = \log \mu_d$ if $a\in (0,a_c^+]$, it follows that
$\varrho(a) = \log \sfrac{\mu_d}{\mu_{d-1}}$ for $a\in (0,a_c^+]$.
In addition, $\varrho (a) \simeq \log a$ if $a\in\RealN$ and $a\to\infty$.
Numerical approximations of $\varrho(a)$ may be made by approximating
the distribution of zeros in the complex plane using the distribution functions
$h_{n}$ (or the distribution functions $f_n$ and $g_n$ 
in equations \Ref{eqn28BB} and \Ref{eqn28CC}).   This gives
numerical approximations to $\varrho(a)$ by 
$\sfrac{1}{n} \sum_{j=1}^n\log (a\minus a_j)$.

The function $\varrho(a)$ is singular at $a=a_c^+$ on the real axis.
This shows that in the limit $n\to\infty$, a singular point develops on the
real axis -- this is the edge singularity.

In the limit as $a \to\infty$, $\log(a\minus z)
= \log a - \sfrac{z}{a} + O(a^{-2})$.  This gives
\begin{equation}
\C{A}^+(a) = \log \mu_{d-1} + \log a - \sfrac{1}{a} \int_{\ComplexN} z\,p(z)\,d\alpha(z)
 + O(a^{-2}),
\label{eqn40A}   
\end{equation}
showing that $\C{A}^+(a)$ is asymptotic to $\log \mu_{d-1} + \log a$ as $a\to\infty$.

Numerical approximations to the distribution $p_n(z)$ is obtained from
equation \Ref{eqn26Z}.  That is,
\begin{equation}
p_n(z) = \sfrac{1}{n} \sum_{J=1}^n h_{j,n}(z) .
\end{equation}
In this event, one may use the distribution functions $f_{j,n}(z)$ or
$g_{j,n}(z)$ as examples.  In the limit as $n\to\infty$, $p_n(a)$ is an approximation
of the limiting distribution $p(z)$.

\subsubsection{Radial and angular distribution of Fisher zeros:}  
The radial and angular distribution of partition function zeros $\LA a_j \RA$
can be considered in view of theorems \ref{thm3.3} and \ref{thm3.6}.  
Define $\nu_n(\rho)$ and $\alpha_n(\theta,\phi)$ similar to the definitions 
in equation \Ref{eqn12},  but now for the zeros of $A_n^+(a)$.
Proceed by computing $L_n$ in theorem \ref{thm3.3}.  By equation \Ref{eqn4}, 
\begin{eqnarray}
L_n &= \log \sum_{v=0}^n c_n^+(v) - \shalf \log c_n^+(0) - \shalf \log c_n^+(n) \cr
&= \log c_n^+ - \shalf \log c_{n-1}^+ - \shalf c_n^{(d-1)}
= \log \Sfrac{c_n^+}{\sqrt{c_{n-1}^+\, c_n^{(d-1)}}} ,
\label{eqn17} 
\end{eqnarray}
since $c_n^+(n) = c_n^{(d-1)}$.  Dividing by $n$ and taking $n\to\infty$ gives the result
\begin{equation}
\zeta_d = \lim_{n\to\infty} \sfrac{2}{n} \, L_n =  \log \Sfrac{\mu_d}{\mu_{d-1}} ,
\label{eqn28} 
\end{equation}
where $\mu_1=1$ (since $c_n^+(n) = 2$ if $n>0$ in $\mathL^2_+$).

Observe that as $d\to\infty$, then the right hand side approaches zero.  Otherwise,
for $d=2$, $\zeta_2 = \lim_{n\to\infty} \sfrac{2}{n} \, L_n \approx 0.970 < 1$
and for $d=3$, $\zeta_3 = \lim_{n\to\infty} \sfrac{2}{n} \, L_n \approx 0.574 < 1$.

By theorem \ref{thm3.3},
\begin{equation}
1- \sfrac{2}{n\rho}\,L_n \leq \sfrac{1}{n}\, \nu_n(\rho) \leq 1 .
\label{eqn19} 
\end{equation}
That is, a positive fraction of the zeros are confined to the annulus 
$1\minus\rho \leq |a| \leq \sfrac{1}{1-\rho}$ in the $a$-plane if the 
left hand side is positive.   Taking $n\to\infty$ gives
\begin{equation}
1- \sfrac{1}{\rho}\,\zeta_d 
\leq \liminf_{n\to\infty} \sfrac{1}{n}\, \nu_n(\rho)
\leq \limsup_{n\to\infty} \sfrac{1}{n}\, \nu_n(\rho) \leq 1 .
\label{eqn20} 
\end{equation}
The left hand side is positive if $\rho > \zeta_d$, and for these values of $\rho$ there is a positive
density of zeros in the annulus $1\minus \rho \leq |a| \leq \sfrac{1}{1-\rho}$.

Since $\liminf_{n\to\infty} \sfrac{1}{n} \, \nu_n(\rho)>0$ if $\rho>\zeta_d$, this gives an
upper bound on the critical point $a_c^+$, namely $a_c^+ \leq \sfrac{1}{1-\zeta_d}$.
This is a very poor bound in the square lattice, namely $a_c^+ \leq 33.42...$, since
$\zeta_2 = \log 2.63815\ldots = 0.97008\ldots$ \cite{CJ12}.  A slightly better bound is
obtained in the cubic lattice, namely $a_c^+ \leq \log \sfrac{4.68404\ldots}{2.63815\ldots}
= 2.347\ldots$ \cite{CJ12,C13}.  Numerical estimates of $a_c^+$ are $1.77564\ldots$
in $\mathL^2_+$  \cite{BGJ12} and $a_c^+=1.306\pm0.007$ in $\mathL^3_+$ \cite{JvR16}.  
In the limit $d\to\infty$ a positive fraction 
of the zeros accumulate on the unit circle in the $a$-plane.

\begin{theorem}
Let $\zeta_d = \log \sfrac{\mu_d}{\mu_{d-1}}$ (and define $\mu_1=1$).  
If $\zeta_d < \rho < 1$, then
\[ 0 < 1- \sfrac{1}{\rho}\,\zeta_d \leq \liminf_{n\to\infty} \sfrac{1}{n}\, \nu_n(\rho)
\leq \limsup_{n\to\infty} \sfrac{1}{n}\, \nu_n(\rho) \leq 1\]
and a positive fraction of partition function zeros of $A_n^+(a)$ are in the annulus
$1\minus\rho < |a| < \sfrac{1}{1-\rho}$ in the complex $a$-plane. \qed
\label{thm2.6}  
\end{theorem}

The angular distribution of partition function zeros can be constrained by 
the Erd\"os-Tur\'an theorem to 
\begin{equation}
|\sfrac{1}{n}\, \alpha_n (\theta,\phi) - \sfrac{1}{2\pi}\, (\phi-\theta)|^2 \leq \sfrac{C_0}{n} L_n,
\end{equation}
with $L_n$ given in equation \Ref{eqn17}, and $C_0$ is an unknown
constant.  This shows, in particular,  that
\begin{equation}
\sfrac{1}{2\pi}\, (\phi-\theta)- \sqrt{\sfrac{1}{n}C_0\,L_n} 
\leq \sfrac{1}{n}\, \alpha_n(\theta,\phi) \leq \sfrac{1}{2\pi}\, (\phi-\theta) 
+ \sqrt{\sfrac{1}{n}\, C_0\,L_n} .
\label{eqn19a} 
\end{equation}
Taking $n\to\infty$ shows that
\begin{eqnarray}
\sfrac{1}{2\pi}\, (\phi-\theta)- \sqrt{\sfrac{1}{2}C_0\,\zeta_d} 
&\leq \liminf_{n\to\infty} \sfrac{1}{n}\, \alpha_n(\theta,\phi) 
\label{eqn20a} 
\\
&\leq \limsup_{n\to\infty} \sfrac{1}{n}\, \alpha_n(\theta,\phi) 
\leq \sfrac{1}{2\pi}\, (\phi-\theta)+ \sqrt{\sfrac{1}{2}\,C_0\,\zeta_d}  .
\nonumber
\end{eqnarray}

\begin{theorem}
Let $\zeta_d = \log \sfrac{\mu_d}{\mu_{d-1}}$.
If $\zeta_d < \rho < 1$, then there exists a constant $C_0$ such that
$$ \limsup_{n\to\infty}  \LV \sfrac{1}{n}\, \alpha_n(\theta,\phi)
-  \sfrac{1}{2\pi}\,(\phi\minus\theta) \RV \leq \sqrt{\sfrac{1}{2}C_0\zeta_d} . \eqno \square $$
\label{cor222}   
\end{theorem}

In low dimensions this is not a strong result, but in the limit as $d\to\infty$,
$\zeta_d \to 0$.  That is, if $C_0\zeta_d \to 0$ as $d\to\infty$, then
$ \limsup_{n\to\infty}  \LV \sfrac{1}{n}\, \alpha_n(\theta,\phi)
-  \sfrac{1}{2\pi}\,(\phi\minus\theta) \RV \to 0$. This does not prove,
but does suggest that the angular distribution of zeros may become more
uniform in higher dimensions.

Since $\limsup_{n\to\infty} \sfrac{1}{n} \, \alpha_n(-\pi,\phi)$ 
is a monotonic function of $\phi$, it is measurable and differentiable almost
everywhere in $(-\pi,\pi]$.  This shows that there exists a measurable function $q$ such that
\begin{equation}
\limsup_{n\to\infty} \sfrac{1}{n} \, \alpha_n(-\pi,\phi)
= \int_{-\pi}^\phi q(\psi)\,d\psi ,
\end{equation}
where the integral is the Lebesgue integral.  The function $q$ is a distribution
function, and 
\begin{equation}
\int_{[\theta,\phi]} dq = 
\int_\theta^{\phi} q(\psi)\, d\psi 
\label{eqn22} 
\end{equation}
is the fraction of zeros with principal argument greater than or equal to $\theta$ and less than 
or equal to $\phi$.

The bounds on the distribution of zeros above (obtained by using 
theorem \ref{thm3.6}, the Erd\"os and Tur\'an theorem) may be slightly improved by
using the theorem of Erd\'elyi  \cite{E08}.  As before, 
let $P(t)$ be a polynomial with zeros $\LA t_1,t_2,\ldots,t_N\RA$.  Define the
following functions related to $\alpha_N(\theta,\phi)$: 
\begin{eqnarray}
\alpha_N^+(\theta,\phi) &= \#
\LC t_k \vert\, \theta < \Arg t_k  \leq \phi,\;\hbox{and}\: |t_k| \geq 1  \RC ; \cr
\alpha_N^-(\theta,\phi) &= \#
\LC t_k \vert\, \theta < \Arg t_k  \leq \phi,\;\hbox{and}\: |t_k| \leq 1 \RC .
\end{eqnarray}

\begin{theorem}[Erd\'elyi  \cite{E08}]
Suppose $\LA a_n\RA$ is a sequence in $\ComplexN$ and suppose that $|a_0a_N| \not=0$.
Define $P(t) = \sum_{n=0}^N a_n \, t^n$ and let $\| P \| = \max_{|t|=1}\{ |P(t)| \}$.  
Then
\[  \sfrac{1}{N}\, \alpha_N^+(\theta,\phi) - \sfrac{1}{2\pi}\,(\phi \minus\theta) 
\leq 16 \sqrt{\sfrac{1}{N}\, \log R_1} ;\]
and
\[  \sfrac{1}{N}\, \alpha_N^-(\theta,\phi) - \sfrac{1}{2\pi}\,(\phi \minus\theta) 
\leq 16 \sqrt{\sfrac{1}{N}\,\log R_2} ,\]
where
$$
 R_1 = |a_N|^{-1} \| P \|,\; \hbox{and} \: R_2 = |a_0|^{-1}\| P \| . \eqno \qed 
$$
\label{thm3.7}   
\end{theorem}

In the particular case here, the polynomial $P(t)$ is the partition function
$A_n^+(a)$.  It follows that $\log R_1 = \log \sfrac{c_n^+}{c_n^{(d-1)}}$ and
$\log R_2 = \log \sfrac{c_n^+}{c_{n-1}^+}$.   Notice that
\begin{eqnarray*}
\lim_{n\to\infty} \sfrac{1}{n} \log R_1 &=  \log \sfrac{\mu_d}{\mu_{d-1}} .
\end{eqnarray*}
Since $c_{n-1}^+ \leq c_n^+ \leq 2d\, c_{n-1}^+$ if follows that 
$1 \leq \limsup_{n\to\infty} \log \sfrac{c_n^+}{c_{n-1}^+} \leq \log 2d$.  In other words,
\begin{eqnarray*}
 \sfrac{1}{n} \log R_2 &= O(\sfrac{1}{n} ) .
\end{eqnarray*}
Put these results together in theorem \ref{thm3.7} to obtain
\[  \sfrac{1}{n}\, \alpha_n^+(\theta,\phi) - \sfrac{1}{2\pi}\,(\phi \minus\theta) 
\leq 16 \sqrt{\sfrac{1}{n}\,\log \sfrac{c_n^+}{c_n^{(d-1)}}} = 
O(1) \q \hbox{as $n\to\infty$};\]
and
\[  \sfrac{1}{n}\, \alpha_n^-(\theta,\phi) - \sfrac{1}{2\pi}\,(\phi \minus\theta) 
\leq 16 \sqrt{\sfrac{1}{n}\,\log \sfrac{c_n^+}{c_{n-1}^+}} =
O(\sfrac{1}{\sqrt{n}}) \longrightarrow 
0\q\hbox{as $n\to\infty$}.\] 
These results give a slightly better outcome, since the bounds are numerical, 
not involving a constant $C_0$ of unknown magnitude as in corollary
\ref{cor222}.  

In the limit as $n\to\infty$, 
\begin{equation}
\limsup_{n\to\infty} \sfrac{1}{n}\, \alpha_n^-(\theta,\phi) \leq
 \sfrac{1}{2\pi}\,(\phi\minus\theta) .
\label{eqn34} 
\end{equation}
Notice that 
\begin{equation}
\lim_{\theta \nearrow \phi} \alpha_n^- (\theta,\phi) = O(\sqrt{n}) .
\end{equation}
In other words, as $n$ increases, the multiplicity of Fisher zeros with $|a|\leq 1$
will increase at a rate no faster than $O(\sqrt{n})$.

Adding the contributions from $\alpha_n^-$ and $\alpha_n^+$ show that, given
a small $\epsilon>0$, there exists an $N$ such that
\begin{equation}
\sfrac{1}{n} \alpha_n(\theta,\phi) \leq
\sfrac{1}{n} \alpha_n^-(\theta,\phi) + \sfrac{1}{n} \alpha_n^+(\theta,\phi)
\leq \sfrac{1}{\pi}\,(\phi \minus\theta) + 16\sqrt{\log \sfrac{\mu_d}{\mu_{d-1}}}  +   \eps .
\end{equation} 
for all $n\geq N$.
This is an improvement on equation \Ref{eqn19a} in that there are no arbitrary
constants $C_0$ involved.  However, this bound is not very useful in low dimensions, but
since the square root term approaches zero with increasing $d$, it gives a useful bound
in high dimensions.

Taking the limsup on the left hand side as $n\to\infty$ gives 
\begin{equation}
\limsup_{n\to\infty} \sfrac{1}{n} \alpha_n(\theta,\phi) \leq
 \sfrac{1}{\pi}\,(\phi \minus\theta) + 16\sqrt{\log \sfrac{\mu_d}{\mu_{d-1}}} .
\end{equation}
This gives the following theorem:

\begin{theorem}
Let $\alpha_n(\theta,\phi)$ be the number of Fisher zeros with principal argument
in $[\theta,\phi]$ in the complex $a$-plane.  Then there exists a measurable function
$q$ defined by
\[ \int_\theta^\phi q(\psi)\, d\psi = \limsup_{n\to\infty} \sfrac{1}{n}\, \alpha_n(\theta,\phi)  ,\]
where $- \pi<\theta\leq\phi \leq \pi$.  The function $q$ is the distribution function of the
principal arguments of the zeros of $A_n^+(a)$ in the limit as $n\to\infty$.  Moreover,
\[\int_{[\theta,\phi]}dq = \int_\theta^\phi q(\psi)\,d\psi \leq \min\{\sfrac{1}{\pi}\,
(\phi\minus\theta)+ 16\sqrt{\log \sfrac{\mu_d}{\mu_{d-1}}},1\},\]
where the integral is the Lebesgue integral on $[-\pi,\pi]$.
\qed
\label{thm2.7a}   
\end{theorem}

\subsection{Yang-Lee zeros}
\label{Yang-Lee}

Lee-Yang zeros were defined in section \ref{Lee-Yang} as $t$-plane zeros of the
partial generating function $G_N(a,t)$ (see equation \Ref{eqn11}).  
In this section the $a$-plane zeros of $G_N(a,t)$
are considered;  these will be called \textit{Yang-Lee} zeros.  Generally, since for fixed
values of $t$ the partial sum  $G_N(a,t)$ is a linear combination of the partition functions
$A_n^+(a)$, the Yang-Lee zeros will have properties similar to the properties
of partition function (or Fisher) zeros. 

Define 
\begin{equation}
B_v(t) = \sum_{n=v}^N c_n^+(v)\, t^n .
\end{equation}
Then the partial sums of the generating function are
\begin{equation}
G_N(a,t) = \sum_{n=0}^N A_n^+(a)\, t^n
= \sum_{v=0}^N \sum_{n=v}^N c_n^+(v) t^n a^v
= \sum_{v=0}^N B_v(t)\, a^v .
\label{eqn36}   
\end{equation}
In the context here, the function $\nu_N(\rho)$ (see equation \Ref{eqn12})
is the number of $a$-plane zeros of 
$G_N(a,t)$ in the annulus $1\minus\rho \leq |a| \leq \sfrac{1}{1-\rho}$.
To apply theorem \ref{thm3.3}, compute $L_N$:
\begin{eqnarray}
L_N 
&= \log \sum_{v=0}^N B_v(t) - \shalf \log B_N(t) - \shalf \log B_0(t) \cr
&= \log \sum_{n=0}^N \sum_{v=0}^n c_n^+(v)\,t^n - \shalf \log c_N^+(N)\, t^N
- \shalf \log \sum_{n=0}^N c_n^+(0)\, t^n \cr
&= \log \sum_{n=0}^N c_n^+\, t^n - \shalf \log c_N^{(d-1)}\,t^N - \shalf \log \sum_{n=1}^N c_{n-1}^+\, t^n \cr
&= \log \sum_{n=0}^N c_n^+\, t^n - \shalf \log c_N^{(d-1)}\,t^N - \shalf \log \sum_{n=0}^{N-1} c_{n}^+\, t^n
-\shalf \log t \cr
&= \log \LH \frac{ \sum_{n=0}^N c_n^+\, t^n}{(\sum_{n=0}^{N-1} c_{n}^+\, t^n )^{1/2}\,
(c_N^{(d-1)}\,t^N)^{1/2}} \RH - \shalf \log t .
\label{eqn38}   
\end{eqnarray}
Suppose that $t$ is large, so that $c_n^+t^n = (\mu_d t)^{n+o(n)} \to\infty$ as $n\to\infty$
(that is, $t > \sfrac{1}{\mu_d}$).  Then the summations above are dominated by their
fastest exponentially growing terms, and the result is that (see equation \Ref{eqn28})
\begin{equation}
\lim_{N\to\infty} \sfrac{2}{N}\, L_N = \log \sfrac{\mu_d}{\mu_{d-1}} = \zeta_d .
\end{equation}
By theorem \ref{thm3.3}, this is similar to equation \Ref{eqn19}, namely, for finite $n$
and $\rho\in(0,1)$
\begin{equation}
1- \sfrac{2}{N\rho}\,L_n \leq \sfrac{1}{N} \nu_N(\rho) \leq 1 + \sfrac{2}{N\rho}\,L_n .
\label{eqn32} 
\end{equation}
Since $\nu_N(\rho) \leq N$, the result is that, in the limit $N\to\infty$,
\begin{equation}
1 - \sfrac{1}{\rho}\zeta_d 
\leq \liminf_{N\to\infty} \sfrac{1}{N}\, \nu_N(\rho) 
\leq \limsup_{N\to\infty} \sfrac{1}{N}\, \nu_N(\rho) 
\leq 1,
\end{equation}
where $\zeta_d = \log \sfrac{\mu_d}{\mu_{d-1}}$.  This result is similar to theorem
\ref{thm2.6}, namely that for $\zeta_d < \rho < 1$, a positive fraction of the 
Yang-Lee zeros are confined to the annulus $1-\rho \leq |a| \leq \sfrac{1}{1-\rho}$
in the $a$-plane.    This gives the following theorem:
\begin{theorem}
Suppose that $t>\sfrac{1}{\mu_d}$.  Define $\zeta_d = \log \sfrac{\mu_d}{\mu_{d-1}}$ 
(where $\mu_1=1$). If $\zeta_d < \rho < 1$, then
\[ 0 < 1- \sfrac{1}{\rho}\,\zeta_d \leq \liminf_{N\to\infty} \sfrac{1}{N}\, \nu_N(\rho) 
\leq \limsup_{N\to\infty} \sfrac{1}{N}\, \nu_N(\rho) \leq 1\]
and a positive fraction of Yang-Lee zeros of $G_N^+(a,t)$ are in the annulus
$1\minus\rho < |a| < \sfrac{1}{1-\rho}$ in the complex $a$-plane. \qed
\label{thm2.7}   
\end{theorem}

Consider next the case that $t$ is small, namely $t< \sfrac{1}{\mu_{d}}< \sfrac{1}{\mu_{d-1}}$.
Then there is a constant $K$ such that $c_n^+ t^n \leq K$ for all $n\in\NatN$.
In particular, $1\leq \sum_{n=0}^N c_n^+ t^n \leq K(N\plus 1)$.  This shows that
$\lim_{N\to\infty} ( \sum_{n=0}^N c_n^+ t^n )^{1/N} = 1$.  Thus, by equation \Ref{eqn38},
\begin{equation}
\xi_t = \lim_{N\to\infty} \sfrac{2}{N} L_N = - \log (\mu_{d-1}t) > 0
\end{equation}
since $t< \sfrac{1}{\mu_{d-1}}$.

Observe that $\xi_t < 1$ if $t\in (\sfrac{1}{e\,\mu_{d-1}},\sfrac{1}{\mu_d})$ and if $\mu_d
< e\, \mu_{d-1}$.  In $\mathL^2_+$ this gives the range $t\in (0.36787\ldots,0.37905\ldots)$,
and in $\mathL^3_+$, $t\in (0.13944\ldots,0.21349\ldots)$.  For these values $t$,
$\xi_t < 1$ and there is a $\rho\in(\xi_t,1)$ for which a positive fraction of zeros are 
confined in an annulus $1\minus\rho < |a| < \sfrac{1}{1-\rho}$ in the complex $a$-plane. 
This gives the theorem:

\begin{theorem}
Suppose that $t\in (\sfrac{1}{e\,\mu_{d-1}},\sfrac{1}{\mu_d})$.  
Define $\xi_t = |\log (\mu_{d-1} t)|$ 
(where $\mu_1=1$). If $\xi_t < \rho < 1$, then
\[ 0 < 1- \sfrac{1}{\rho}\,\xi_t \leq \liminf_{N\to\infty} \sfrac{1}{N}\, \nu_N(\rho)
\leq \limsup_{N\to\infty} \sfrac{1}{N}\, \nu_N(\rho) \leq 1\]
and a positive fraction of Yang-Lee zeros of $G_N^+(a,t)$ are in the annulus
$1\minus\rho < |a| < \sfrac{1}{1-\rho}$ in the complex $a$-plane. \qed
\label{thm2.8}   
\end{theorem}

Theorems \ref{thm2.7} and \ref{thm2.8} do not rule out the possibility that zeros may accumulate in
a annular region for $0<t\leq \sfrac{1}{e\,\mu_{d-1}}$.

The angular distribution of Yang-Lee zeros can be considered using theorem \ref{thm3.7}.
In this case, the polynomial is $G_N(a,t) = \sum_{v=0}^N B_v(t)\, a^v$
(equation \Ref{eqn36}).   Similar to the above, let $\alpha_N(\theta,\phi)$ be defined 
as in equation \Ref{eqn12}, but now for Yang-Lee zeros of the partial generating function
$G_N(a,t)$ in the $a$-plane.  Proceed by computing $R_1$ and $R_2$ in theorem \ref{thm3.7}:  
First, notice that 
\[ \hspace{-2cm}
\| G_N(a,t) \| = \max_{|a|=1}\LC\sum_{v=1}^N B_v(t)\, a^v \RC
= \sum_{v=0}^N B_v(t)
= \sum_{v=0}^N \sum_{n=v}^N c_n^+(v)\, t^n
= \sum_{n=0}^N c_n^+\,t^n. \]
This gives for $R_1$:
\begin{equation}
R_1 = B_N(t)^{-1} \| G_N(a,t) \| = \frac{\sum_{n=0}^N c_n^+ \, t^n}{c_N^+(N)\, t^N}
= \frac{\sum_{n=0}^N c_n^+ \,t^n}{c_N^{(d-1)}\,t^N} .
\end{equation}
Taking logs, dividing by $N$, and taking $N\to\infty$, shows that
\begin{equation}
\lim_{N\to\infty} \sfrac{1}{N}\, \log R_1
= \cases{
\log \sfrac{\mu_d}{\mu_{d-1}}, & \hbox{\norf if $t>\sfrac{1}{\mu_d}$}; \cr
- \log (\mu_{d-1}t) & \hbox{\norf if $t\leq \sfrac{1}{\mu_d}$}.
}
\label{eqn44}   
\end{equation}
$R_2$ is computed in a similar way, namely
\begin{equation}
R_2 = B_0(t)^{-1} \| G_N(a,t) \| = \frac{\sum_{n=0}^N c_n^+ \, t^n}{c_N^+\, t^N} .
\label{eqn45}   
\end{equation}
Taking logs, dividing by $N$, and taking $N\to\infty$, shows that
\begin{equation}
\lim_{N\to\infty} \sfrac{1}{N}\, \log R_2
= 0 ,
\end{equation}
for all $t>0$, since the numerator and denominator grows at the same exponential 
rate in equation \Ref{eqn45} if $t>\sfrac{1}{\mu_d}$.
Since the result in equation \Ref{eqn44} is positive for allmost all $t>0$, it follows by
theorem \ref{thm3.7} that 
\begin{eqnarray}
\limsup_{N\to\infty} \sfrac{1}{N}\, \alpha_N^+(\theta,\phi) - \sfrac{1}{2\pi}\,(\phi\minus\theta)
\leq 16 \cases{
\sqrt{\log\sfrac{\mu_d}{\mu_{d-1}}}, & \hbox{\norf if $t>\sfrac{1}{\mu_d}$}; \cr
\sqrt{|\log (\mu_{d-1}t)|},  & \hbox{\norf if $t\leq\sfrac{1}{\mu_d}$}.
} \\
\limsup_{N\to\infty} \sfrac{1}{N}\, \alpha_N^-(\theta,\phi) - \sfrac{1}{2\pi}\,(\phi\minus\theta)
\leq 0 . \nonumber
\end{eqnarray}
Adding these results give
\begin{equation}
\limsup_{N\to\infty} \sfrac{1}{N}\, \alpha_N(\theta,\phi) 
\leq \sfrac{1}{\pi}\,(\phi\minus \theta) + 16 \cases{
\sqrt{\log\sfrac{\mu_d}{\mu_{d-1}}}, & \hbox{\norf if $t>\sfrac{1}{\mu_d}$}; \cr
\sqrt{|\log (\mu_{d-1}t)|},  & \hbox{\norf if $t\leq\sfrac{1}{\mu_d}$}.
}
\end{equation}
This is not a particularly good bound on the limsup, since the factor of $16$
is quite significant.  It is also the case that $\lim_{N\to\infty} \sfrac{1}{N}\, \alpha_N(\theta,\phi) 
\leq 1$, and since $16\sqrt{\log \mu_2} \approx 15.76$ and
$16\sqrt{\log \sfrac{\mu_3}{\mu_{2}}} \approx 12.12$, this bound
is not useful in $\mathL^d_+$ if $d=2$ or $d=3$.  However, taking $d$ large
has the consequence that $\sqrt{\log\sfrac{\mu_d}{\mu_{d-1}}} \to 0$,
so that the bound becomes far better, and  approaches the bound given
in equation \Ref{eqn34} for Fisher zeros.

\begin{figure}[t!]
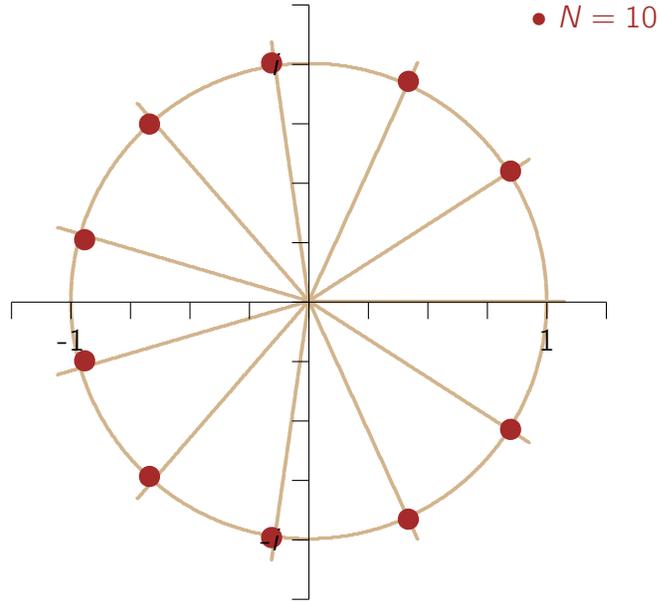

\input Figures/figure02.tex
\caption{Approximate locations of the zeros of $g_N(a,t)$ for $N=10$
and $a=1$.  The zeros are approximately on the unit circle, and close
to ten vertices of a regular 11-gon with centre at the origin (and with one
corner placed at the point $(1,0)$, as illustrated).  The rays from the origin
are the spokes of the regular 11-gon.  The least squares ellipse
through the points has approximate centre $(0.011,0)$, approximate
half short axis of length $1.00$, and approximate half long axis of
length $1.01$.}
\label{figure-t2-1}  
\end{figure}

Similar to the case for Fisher
zeros, there exists an angular distribution function $s(\psi)$ for Yang-Lee zeros.
This distribution function is similarly constrained by the results above in 
high dimensions, as shown in the following theorem:

\begin{theorem}
Let $\alpha_N(\theta,\phi)$ be the number of Yang-Lee zeros with principal argument
in $[\theta,\phi]$ in the complex $a$-plane.  Suppose that $s$ is defined by
\[ \int_\theta^\phi s(\psi)\, d\psi = \limsup_{n\to\infty} \sfrac{1}{N}\, \alpha_N(\theta,\phi)  ,\]
where $- \pi<\theta\leq\phi \leq \pi$.  Then $s$ is a distribution function of the
principal arguments of the $a$-plane zeros of $G_N(a,t)$ in the limit as $n\to\infty$.  Moreover,
if $t>\sfrac{1}{\mu_d}$, then
\[\int_{[\theta,\phi]}ds = \int_\theta^\phi s(\psi)\,d\psi \leq \min\{\sfrac{1}{\pi}\,
(\phi\minus\theta)+16\sqrt{\log\sfrac{\mu_d}{\mu_{d-1}}},1\},\]
where the integral is the Lebesgue integral on $[-\pi,\pi]$.

Similarly, if $t<\sfrac{1}{\mu_d}$, then 
\[\int_{[\theta,\phi]}ds = \int_\theta^\phi s(\psi)\,d\psi \leq \min\{\sfrac{1}{\pi}\,
(\phi\minus\theta)+16\sqrt{|\log(\mu_{d-1}t)|},1\}. \hfill \qed \]
\label{thm2.7aa}   
\end{theorem}

\section{Numerical determination of the distribution of zeros}
\label{Numerical}   

In this section the location and distribution of Fisher zeros are examined, using
in particular the measures $\mu_n$ and $\nu_n$ in equations \Ref{eqn10q}
and \Ref{eqn28z}.  Similar comments can be made about Lee-Yang zeros,
but since these are, from a numerical point of view, very stably distributed
along a circle in the $t$-plane, it appears that numerical error is not a concern
here.  For example, the zeros of $g_N(a,t)$ for $N=10$ and $a=1$ are shown in
figure \ref{figure-t2-1}.  These are located very close to the vertices of a 
$11$-gon on the unit circle in the $t$-plane (with no zero close to the
vertex on the positive real axis).  Similar results will be presented in sections
\ref{Lee-Yang2} and \ref{Lee-Yang3}.

\subsection{Numerical location of Fisher zeros}

Numerical estimates of $c_n^+(v)$ in $\mathL^2_+$ and $\mathL^3_+$ were
obtained by using the GAS algorithm \cite{JvRR09} to approximately enumerate 
self-avoiding walks \cite{JvR10} in half lattices.  In each case the algorithm was
used to sample along $10^3$ sequences each of length $10^9$ iterations
(see reference \cite{JvR16} for details), estimating microcanonical data for
$0\leq v \leq n$ where $0\leq n \leq 500$ over $10^{12}$ iterations.  
Since the algorithm estimates approximate values $\overline{c}_n^+(v)$
for $c_n^+(v)$, the numerically determined partition functions 
and partial generating functions are only approximations 
to $A_n^+(a)$ and $G_N(a,t)$ in $\mathL^2_+$ and $\mathL^3_+$.

An error in the estimate of $c_n^+(v)$ will change the locations of partition 
and partial generating function zeros.  The degree to which this is the case
may be examined by assuming a small relative error of $\epsilon$ in one of the 
estimates (so that $\overline{c}_n^+(v) = (1+\epsilon) c_n^+(v)$), and then to
examine the impact this has on the locations of zeros.

\begin{figure}[t]
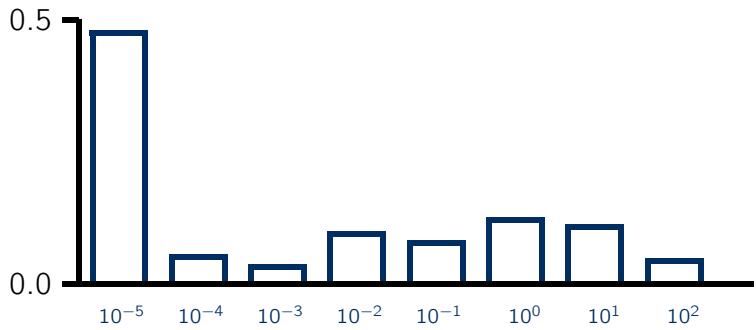

\input Figures/figure03.tex
\caption{A histogram showing the distribution of $|c_n^+(k)\thin (a_j^*)^k|/|A_n^{+\prime}(a_j)|$
for $n=400$ in the half square lattice.  Bars with label $10^{-n}$ represent the 
fraction of ratios of size larger than $10^{-n-1}$ and smaller or equal to
$10^{-n}$.  Almost all the ratios are smaller than $10^{-5}$, and less than $10\%$ of
the ratios are bigger than $1$.  The sizes of those ratios exceeding $1$ are still well smaller
than $10^2$.}
\label{histogram} 
\end{figure}

Consider the partition function $A_n^+(a)$ and suppose that the coefficient
of $a^k$ has a small relative error (that is, $\overline{c}_n^+(k) = (1+\epsilon) c_n^+(k)$).
Define the perturbed partition function $P_n(a)$ by
\begin{equation}
P_n(a) = A_n^+(a) + \epsilon\, c_n^+(k)\, a^k .
\end{equation}
Suppose that the $j$-th zero of $A_n^+(a)$ is $a_j$, and that the zeros of $P_n(a)$
are denoted by $\langle a_j^* \rangle$.  Then there is a labeling of the zeros
so that $a_j^* = a_j$ for each $j$ if $\eps=0$.  If $|\eps|>0$ is very small, then assume
that for the $j$-th zero, $a_j^* = a_j + \delta_j$ for some value of $\delta_j$.   The size
of $\delta_j$ will be estimated as follows:  Consider
\begin{eqnarray*}
P_n(a_j^*) &= A_n^+ (a_j^*) + \eps\, c_n^+(k) (a_j^*)^k \\
&= A_n^+(a_j) + \delta_j\, A_n^{+\prime}(a_j) + O(\delta_j^2) +  \eps\, c_n^+(k) (a_j^*)^k,
\end{eqnarray*}
where it is assumed that $A_n^{+\prime}(a_j) \not=0$.  Since $P_n(a_j^*) = A_n^+(a_j)=0$,
solve for $\delta_j$ to leading order in $\eps$:
\begin{equation}
\delta_j = \frac{\eps\, c_n^+(k)\thin (a_j^*)^k}{A_n^{+\prime}(a_j)} + O(\eps^2).
\label{eqn70AA}
\end{equation}
This shows that $\delta_j \to 0$ as $\eps\to 0$.   Since the relative error $\eps$ approaches
zero as the GAS algorithm continues to run; it follows that $\delta_j$ decreases in 
size with the length of the simulation, and the estimates of locations of zeros improve.

\begin{figure}[t]
\centering
\begin{minipage}{0.30\textwidth}
\centering
\fbox{\includegraphics[width=0.75\linewidth, height=0.15\textheight]{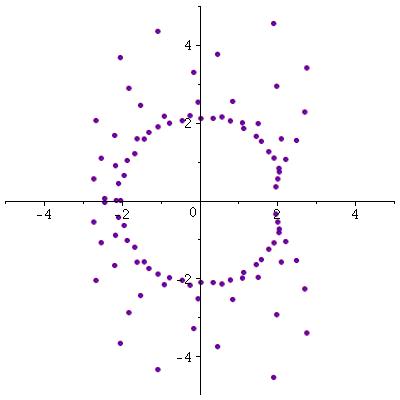}}
\end{minipage}
\begin{minipage}{0.30\textwidth}
\centering
\fbox{\includegraphics[width=0.75\linewidth, height=0.15\textheight]{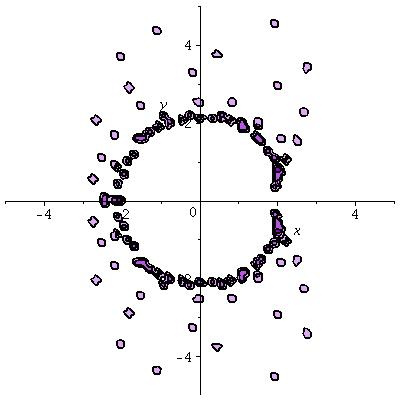}}
\end{minipage}
\begin{minipage}{0.30\textwidth}
\centering
\fbox{\includegraphics[width=0.75\linewidth, height=0.15\textheight]{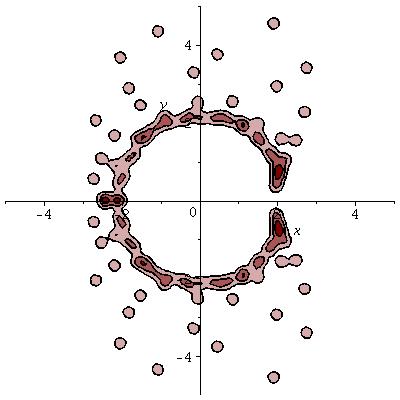}}
\end{minipage}
\caption{Left: Fisher zeros for $n=100$ in $d=2$ dimensions.  The distribution of zeros
can be estimated from this data using the measures $\mu_n$ and $\nu_n$.  Middle:  The
distribution of Fisher zeros using the measure $\mu_n$ and the data on the Left.
Right: The distribution using the measure $\nu_n$ and the data on the Left.}
\label{figure3}
\end{figure}

It is not possible to estimate $\eps$ in the above, except in some cases were $c_n^+(k)$ is
known.  For example, $c_n^+(n) = 2$ in the half square lattice for all $n>0$.  The GAS algorithm 
in our simulations converged to $c_{400}^+(400) \approx 1.9772$, giving
$|\eps| \lesssim 0.012$.  Similarly, for $n=500$ the result is 
$c_{500}^+(500) \approx 1.9751$, so that $|\eps| \lesssim  0.013$.  Incidentally, these
estimates show that the GAS algorithm samples well into the tails of the distribution
$c_n^+(v)$ in $(n,v)$, and its estimates of the microcanonical partition function in the
tails of the distribution over state space is very good -- this is an example of
rare event sampling (and is a quality this algorithm shares with the PERM and flatPERM
algorithms \cite{G97,PK04}).

In three dimensions $c_n^+(n)$ is the number of square lattice walks.  If $n=50$
then the estimate obtained by the GAS algorithm is
$c_{50}^+(50) \approx 5.29279\times 10^{21}$ while exact enumeration \cite{J04}
gives $c_{50}^+(50) = 5.30324\ldots \times 10^{21}$.  This gives $|\eps| \lesssim 0.0020$.
A similar comparison for $n=70$ gives $\eps \lesssim 0.0050$.

For other values of $n$ and $v$ the relative error $\epsilon$ can be estimated by
dividing the standard error of the estimates by the estimate.  That is, if the estimate
$\overline{c}^+_n(v)$ has standard error $e^+_n(v)$, then $\eps$ is estimated by the ratio
$e^+_n(v)/\overline{c}^+_n(v)$.  This gives $\eps\lesssim 0.01$ consistently in 
$\mathL^+_2$ and $\mathL^+_3$, for $n\leq 450$, and for $n>450$ it grows to a couple
of percent for a few values of $v$ (but otherwise remains well below 1\%).

These results suggest that, generally, one may assume that $\eps \lesssim 0.01$ (or that the
microcanonical data were estimated to a relative error of about 1\% for most values of 
$0 \leq v \leq n$ and $n\leq 500$ and for most values of $(n,v)$ in fact far better than this).

The size of $\delta_j$ is also a function of the size of the ratios 
$|c_n^+(k)\thin (a_j^*)^k|/|A_n^{+\prime}(a_j)|$.  Summing equation \Ref{eqn70AA} over
$k$ gives an estimate of the total contribution if $\eps$ is a fixed error on all the data.  
Since $\delta_j$ is a complex number, this gives the total estimated error
\begin{equation}
\Delta_j = \eps\; \frac{A_n^+(a_j^*)}{A_n^{+\prime}(a_j)} + O(n\eps^2) .
\end{equation}
Estimates of his quantity (with $a_j$ replaced by $a_j^*$ in our data) are always very small, 
and never exceeded $10^{-5}$.  However, this is very likely 
an underestimate of the error in the location of the $a_j^*$.

An alternative approach, which should give a large overestimate of the error, is to estimate
$|c_n^+(k)\thin (a_j^*)^k|/|A_n^{+\prime}(a_j)|$ and sum these contributions
over $k$.  The size distribution of $|c_n^+(k)\thin (a_j^*)^k|/|A_n^{+\prime}(a_j)|$
is displayed for $n=400$ in figure \ref{histogram}.  
Notice the logarithmic scale on the horizontal axis of the histogram.  The height
of each bar is proportional to the fraction of ratios
which has size less than the power of 10 shown
at the bottom of the bar.  Summing
the contributions give, in all cases, a result between 10 and 100 (and consistently
well less than 100).  Since this is
the result if all the magnitudes of the errors are added up without cancellations, 
multiplying this by $\eps$ likely gives an overestimate of the error in the location of the $a_j^*$.

If one accepts that the contribution of the ratio 
$\sfrac{c_n^+(k)\thin (a_j^*)^k}{A_n^{+\prime}(a_j)}$ to lie between $10^{-5}$ 
and $10^2$, and $\eps\lesssim 10^{-2}$, then $\delta_j$ will be very small for
most values of $j$, and large for a few values of $j$.  These results suggest 
that a few zeros will move quickly with uncertainties in the data, and 
equation \Ref{eqn70AA} indicates that this will be particularly true for cases where
$|a_j^*|$ is large.  Most zeros, however, should be quite stable. 

Similar observations apply to the $a$-plane zeros of the partial generating function
$G_N(a,t)$.

\begin{figure}[t]
\centering
\begin{minipage}{0.3\textwidth}
\centering
\includegraphics[width=1.125\linewidth, height=0.225\textheight]{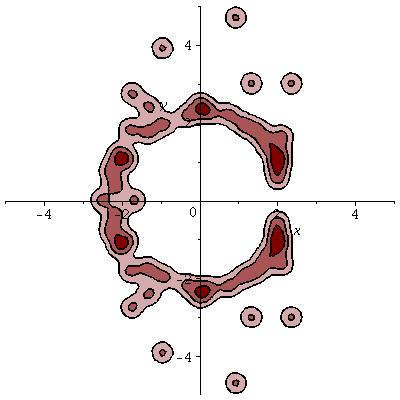}
\end{minipage}
\begin{minipage}{0.30\textwidth}
\centering
\includegraphics[width=1.125\linewidth, height=0.225\textheight]{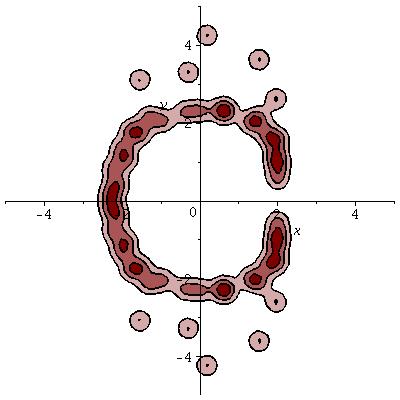}
\end{minipage}

\begin{minipage}{0.30\textwidth}
\centering
\includegraphics[width=1.125\linewidth, height=0.225\textheight]{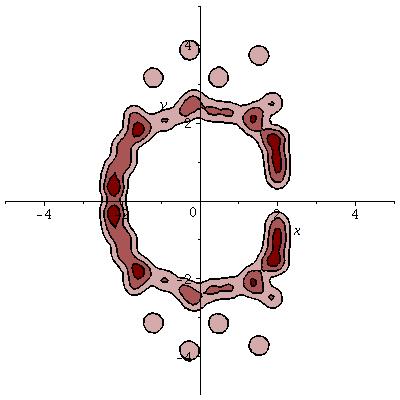}
\end{minipage}
\begin{minipage}{0.30\textwidth}
\centering
\includegraphics[width=1.125\linewidth, height=0.225\textheight]{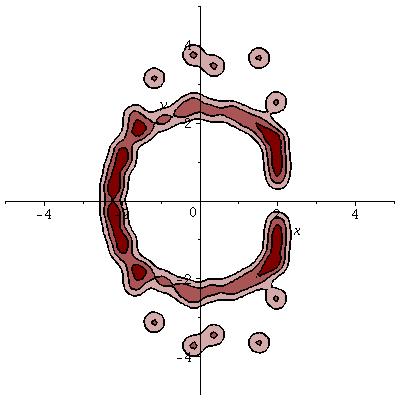}
\end{minipage}
\caption{The evolution of the distribution $p_n$ for $a$-plane zeros for $n=50$ and
sampling along $M$ GAS-sequences of length $10^9$.   Top Left:  The distribution for
$M=10$ ($10$ sequences of length $10^9$ each).  Top Right: the distribution
for $M=100$.  Bottom Left: $M=250$, and Bottom Right: $M=400$.  These
contour plots, computed by using the locations of the zeros and the measure
$\nu_n$, show remarkable consistency with increasing $N$ as the zeros accumulate
in a pattern on a circle of radius equal to the critical value of $a$ (the critical adsorption
point $a_c^+$) in the half square lattice.  Notice the formation of the edge singularity
on the positive real axis.}
\label{figure4}
\end{figure}

\begin{figure}[t]
\centering
\begin{minipage}{0.30\textwidth}
\centering
\includegraphics[width=1.125\linewidth, height=0.225\textheight]{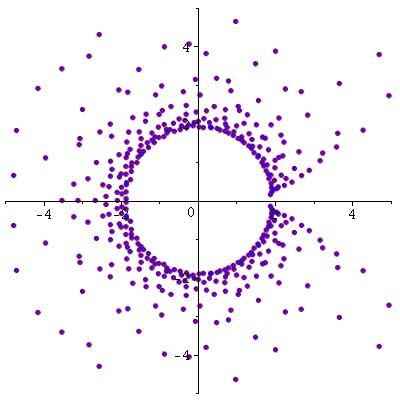}
\end{minipage}
\begin{minipage}{0.30\textwidth}
\centering
\includegraphics[width=1.125\linewidth, height=0.225\textheight]{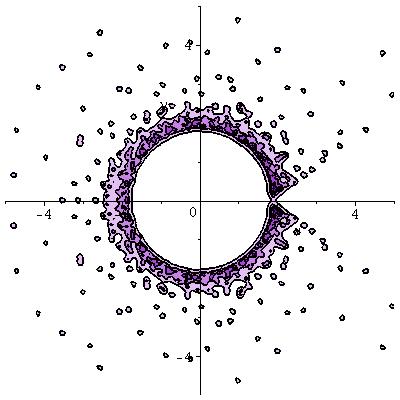}
\end{minipage}
\begin{minipage}{0.31\textwidth}
\centering
\includegraphics[width=1.125\linewidth, height=0.225\textheight]{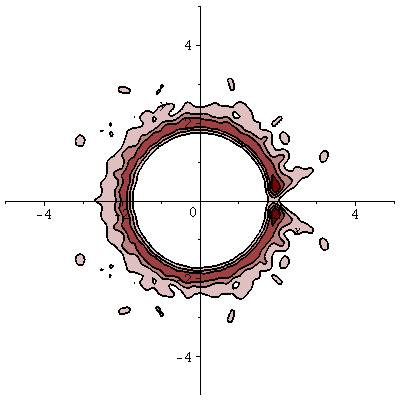}
\end{minipage}

\begin{minipage}{0.30\textwidth}
\centering
\includegraphics[width=1.125\linewidth, height=0.225\textheight]{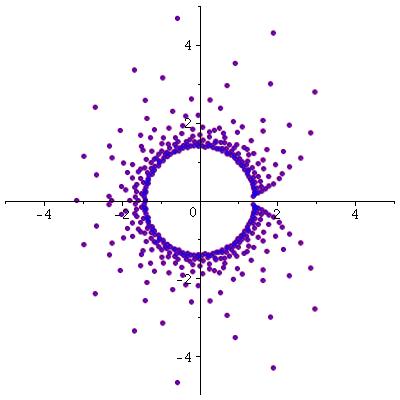}
\end{minipage}
\begin{minipage}{0.30\textwidth}
\centering
\includegraphics[width=1.125\linewidth, height=0.225\textheight]{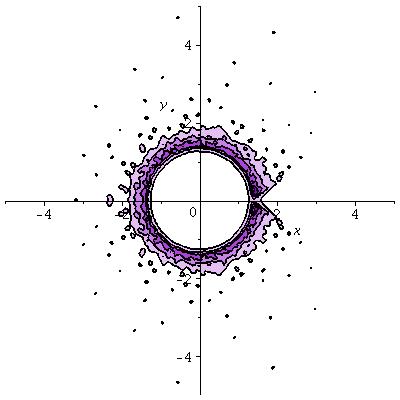}
\end{minipage}
\begin{minipage}{0.30\textwidth}
\centering
\includegraphics[width=1.125\linewidth, height=0.225\textheight]{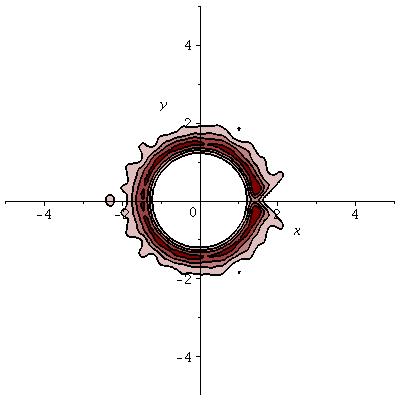}
\end{minipage}
\caption{Fisher zeros of square lattice adsorbing walks of length $n=400$ (top row), and
of cubic lattice adsorbing walks of length $n=400$ (bottom row).  Top left is the locations of the
calculated zeros, top middle a contour plot of the distribution of zeros computed from using the 
measure $\mu_n$ in equation \Ref{eqn10q} with $\tau_n=n$, and top right a contour plot of the distribution
of zeros computed by using the measure $\nu_n$
with $\tau_n = \frac{1}{2}\sqrt{n}$ in equation \Ref{eqn28z}.
Along the bottom row the same plots are shown, but now for data for cubic lattice walks
of length $n=400$.}
\label{figure4B}
\end{figure}

\subsection{The numerical distribution of Fisher zeros}

The distribution of Fisher zeros were measured in section \ref{Dist-Fisher} using
the probability measure $\lambda_n$.  Two examples
of measures were given in equations \Ref{eqn10q} and \Ref{eqn28z}, namely the
measures $\mu_n$ and $\nu_n$.  These measures give distributions
functions $p_n$ on the $a$-plane zeros.  In this section this is considered from a numerical
point of view, using the approximate enumeration data generated by the GAS algorithm.

In figure \ref{figure3} the locations of approximate $a$-plane zeros
of $A_n^+(a)$ for $n=100$ and $d=2$ is shown.  On the left are the locations of the zeros,
and in the middle and on the right are contour plots of the distribution $p_n$ of zeros using the
measures $\mu_n$ (with $\tau_n = n$;
see equation \Ref{eqn10q}) and $\nu_n$ 
(with $\tau_n=\sfrac{1}{2}\sqrt{n}$;
see equation \Ref{eqn28z}), respectively.

{\renewcommand{\baselinestretch}{1.2}
\begin{table}[t!]
\begin{center}
{
\caption{$\|p_n\minus p_{500})\|_1$ with respect to $\mu_n$ and $\nu_n$ in $\mathL_+^2$}
\label{Table1}  
}
{
\begin{tabular}{l|c|c}
\hline\hline
$n$ & $\|p_n\minus p_{500}\|_1$ wrt $\mu_n$ ($\tau_n = \frac{n}{9}$)
      &  $\|p_n\minus p_{500}\|_1$ wrt $\nu_n$ ($\tau_n = \frac{1}{2}\sqrt{n}$) \\
\hline
$50$ & $1.113$ & $1.028$ \\
$100$ & $0.948$ & $0.775$  \\
$200$ & $0.869$ & $0.581$ \\
$300$ & $0.813$ & $0.595$ \\
$350$ & $0.756$ & $0.427$ \\
$400$ & $0.735$ & $0.399$ \\
$420$ & $0.678$ & $0.363$ \\
$440$ & $0.723$ & $0.390$ \\
$460$ & $0.700$ & $0.373$ \\
$470$ & $0.689$ & $0.379$ \\
$480$ & $0.659$ & $0.345$ \\
$490$ & $0.648$ & $0.339$ \\
$495$ & $0.537$ & $0.288$ \\
\hline\hline
\end{tabular}
}
\end{center}
\end{table}
}

The stability of the distribution $p_n$ (or its dependence on uncertainties in the numerical
data), may be examined by tracking the distribution of the estimated zeros as the length 
of the similation increases.  The result is shown in figure \ref{figure4} for $p_n$ in the 
half square lattice.  The first contour plot (top, left) is the distribution $p_n$, determined by using 
the measure $\nu_n$ (with $\tau_n =\sfrac{1}{2}\sqrt{n}$;
see equation \Ref{eqn28z}).  In this case the data were obtained by
using the GAS algorithm to sample along 10 sequences, each of length $10^9$ iterations.

Increasing the length of the simulation to 100 sequences of length $10^9$
gives the second contour plot (top, right).  While the contours have changed, the distribution
is still recognizable.  Notice that those isolated zeros far from the origin moved from their
locations.  The distributions can be compared by computing the $L_1(\alpha)$-norm  
(which has a maximum of $2$) of the absolute difference between the distributions 
$p_n$ for simulations of $N=10$ and $N=100$ sequences (and for walks of
length $n=50$).   With respect to $\mu_n$ with $\tau_n=\sfrac{1}{9}n$ this gives
$0.634$ and with respect to $\nu_n$ with $\tau_n=\sfrac{1}{2}\sqrt{n}$, 
the $L_1(\alpha)$-norm is $0.596$.

Increasing the length of the simulation to 250 sequences of length $10^9$ iterations
gives the distribution on the bottom left.  The contours changed, but the overall shape of the
distribution remained similar.  Computing the $L_1(\alpha)$-norm of the
absolute difference between the distributions at $N=100$
and $N=250$ gives $0.398$ with respect to $\mu_n$
($\tau_n=\sfrac{1}{9}n$) and $0.354$ with respect to $\nu_n$
($\tau_n=\sfrac{1}{2}\sqrt{n}$), respectively.  

For $N=400$ the distribution is given by the contour plot on the bottom right.  
The $L_1(\alpha)$-norm of this distribution and the distribution at $N=250$ is $0.195$ 
with respect to $\mu_n$ and $0.181$ with respect to $\nu_n$, respectively.   Letting
the simulation run to $N=1000$ sequences each of length $10^9$ gives the final distribution
illustrated in figure \ref{figure4B} (in this case illustrated for $n=400$).   The formation 
of an edge-singularity on the positive real axis is seen as the leading zeros approach the
positive real axis.   In this figure the contour plots in the middle and on the right were
determined using the measures $\mu_n$ with $\tau_n=n$, and $\nu_n$
with $\tau_n=\sfrac{1}{2}\sqrt{n}$ in equations \Ref{eqn10q} and \Ref{eqn28z}
respectively. 

The convergence of $p_n \to p$ as $n\to\infty$ can be tracked by computing
the $L_1(\alpha)$-norm of $(p_n\minus p)$ as a function of $n$ (while taking $p\approx
p_{500}$ as the best estimate of the limiting distribution).  This is tabulated in
table \ref{Table1}.  Notice that increasing $n$ decreases the $L_1(\alpha)$-norm in
each case, showing that the distribution $p_n$ approaches $p_{500}$
systematically with increasing $n$.  This is strong supporting evidence that the
distribution of zeros produced by the algorithm behave consistently with 
increasing $n$. Similar results are seen in $\mathL_+^3$.

The relation between the distributions obtained from the $\mu_n$ and $\nu_n$
measures can also be examined by computing the $L_1(\alpha)$-norm of the 
difference between the distributions.  If $p_n(\mu_n)$ is the distribution obtained
from the $\mu_n$-measure, and $p_n(\nu_n)$ is the distribution obtained from the
$\nu_n$-measure, then $\| p_n(\mu_n) \minus p_n(\nu_n)\|_1 = 0.363$ if $n=50$
and $\tau_n = \sfrac{n}{9}$ for $\mu_n$ and $\tau_n = \sfrac{1}{2}\sqrt{n}$
for $\nu_n$.  For $n=500$, the distance is slightly larger at
$\| p_n(\mu_n) \minus p_n(\nu_n)\|_1 = 0.517$.

\section{Numerical Results in the Square Lattice}
\label{Numerical2}   

\subsection{Lee-Yang zeros of $G_N(a,t)$ in the square lattice}
\label{Lee-Yang2}   

The $t$-plane zeros of $G_N(a,t)$ are in general functions of $a$ (and of $N$), but
will be converging on the circle $|t|=\lambda_a^{-1}$ in the $t$-plane 
as $N\to\infty$ by theorem \ref{thm3.4} (where $\lambda_a$ is defined by 
equation \Ref{eqn13}).   Moreover, by theorem \ref{thm3.5} the complex arguments of 
the zeros should approach a uniform distribution.  That is, the 
$t$-plane zeros of $G_N(a,t)$ should be distributed, in the limit, uniformly on 
a circle of radius $\lambda_a^{-1}$. 

\begin{figure}[t]
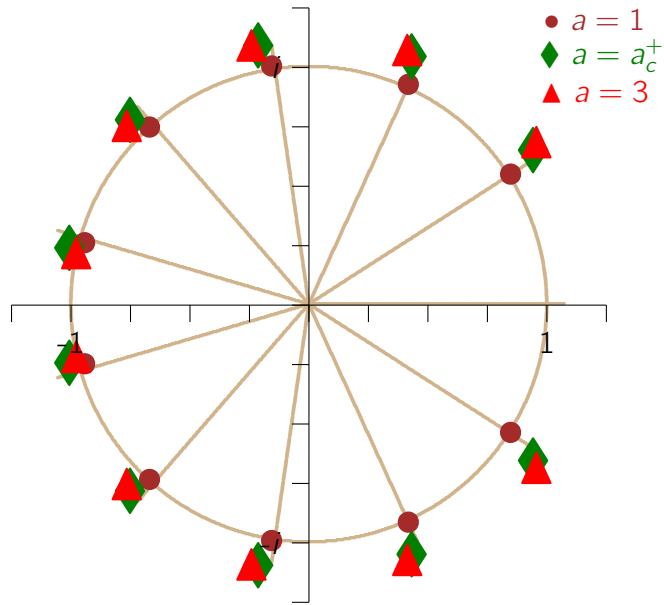

\input Figures/figure07.tex
\caption{Approximate locations of the zeros of $g_N(a,t)$ for $N=10$
and for $a=1$ ($\bullet$), $a=1.779\approx a_c^+$ ($\blacklozenge$),
and $a=3$ ($\blacktriangle$).  The zeros are close to the unit circle, and close
to ten vertices of a regular 11-gon with centre at the origin, and one
corner at the point $(1,0)$, as illustrated.  The rays from the origin
are the spokes of a regular 11-gon.}
\label{figure-t2-3}  
\end{figure}

To investigate the above numerically, approximate  $\lambda_a$ by
\begin{equation}
\lambda_a(N) = (A_N^+(a))^{1/N} . 
\end{equation} 
In particular, $\lim_{N\to\infty} \lambda_a(N) = \lambda_a$.  The location 
of zeros of $G_N(a,t)$ in the $t$-plane can be scaled with $\lambda_a(N)$
in order to approximate the zeros of $g_N(a,t)$ (see equation \Ref{eqn15}).  
The approximation will become increasingly accurate with increasing $N$. 

In figure \ref{figure-t2-1} the approximate zeros of $g_{10}(1,t)$
are plotted in the $t$-plane.  This shows that, even for this small value of $N$,
the estimated zeros are distributed close to the unit circle and spaced 
in a way consistent with approaching a uniform distribution in the limit 
as $N\to\infty$. 

The location of the zeros in figure \ref{figure-t2-1} can be examined
by fitting a least squares ellipse to the data points.  This has approximate centre 
at the point $(0.011,0)$, close to the origin, approximate half short axis
of length $1.00$ and approximate half long axis of length $1.01$.  The rays
from the origin in figure \ref{figure-t2-1} are the spokes of a regular 11-gon 
and they intersect the unit circle in the vertices of a regular $11$-gon; each intersection 
is close to a zero of $g_{10}(1,t)$, except at the point $(1,0)$ on the real axis
where there is no zero nearby. 

The locations of zeros of $g_N(a,t)$ are, in general, functions of $a$
(and of $N$).  This is examined in figure \ref{figure-t2-3} for $N=10$
and three different values of $a$, namely $a=1$ (the same data as 
in figure \ref{figure-t2-1}), $a=1.779\approx a_c^+$, and $a=3$
(at this value the model is in the adsorbed phase).  
There are small movements in the
zeros with increasing $a$, and this can be collectively examined by
determining the least squares ellipse for each value of $a$.  For $a=1$
the least squares ellipse was given above (approximate centre at
$(0.011,0)$, approximate half short axis of length $1.00$ and approximate
half long axis of length $1.01$).  The ellipse for $a=1.779$
has values $(0.049,0)$, $1.125$ and $1.139$ for its approximate
centre, half short axis, and half long axis.  For $a=3$ the values are
$(0.077,0)$, $1.137$ and $1.168$ for its approximate
centre, half short axis, and half long axis.

\begin{figure}[t] 
\centering
\begin{minipage}[t]{0.40\linewidth}
       \flushleft
        \input Figures/figure08.tex
        \caption{Approximate locations of the zeros of $g_N(a,t)$ for $N=100$
and $a=1$.  The zeros are approximately on the unit circle, and close
to 100 vertices of a regular 101-gon with centre at the origin, and with one
corner at the point $(1,0)$, as shown.  
}
\label{figure-t2-2}  
    \end{minipage}%
    \begin{minipage}[t]{0.40\linewidth}
        \flushright
        \input Figures/figure09.tex
        \caption{Approximate locations of the zeros of $g_N(a,t)$ for $N=100$
and for $a=1$, $a=1.779\approx a_c^+$, and $a=3$.  All these zeros are 
located very close to the unit circle.  
}
\label{figure-t2-4}  
    \end{minipage}%
\end{figure}

The situation is similar for larger values of $N$.  For example, the cases for
$N=100$ are illustrated in figures \ref{figure-t2-2} and \ref{figure-t2-4}.
In figure \ref{figure-t2-2} the zeros of $g_N(a,t)$ with $N=100$ and
$a=1$ are shown.  These zeros are very close to $100$ vertices of a regular
101-gon on the unit circle with one vertex at $(1,0)$ (and this is the one
vertex not close to a zero of $g_N(a,t)$).
The least squares ellipse through the zeros has approximate centre 
$(0.0014,0)$, approximate half short axis $1.001$ and approximate half 
long axis $1.002$. 

In figure \ref{figure-t2-4} the zeros are displayed for larger values of $a$.
There are very little motions of the zeros with increasing $a$.  For example, for
$a=1.779 \approx a_c^+$ the least squares ellipse through the zeros 
has approximate centre $(0.0068,0)$, half short and half long axes 
of approximate lengths $1.012$ and $1.016$.  Similarly,
for $a=3$, the least squares ellipse has approximate centre at $(0.0101,0)$,
and half short and half long axes of approximate lengths $1.013$ and $1.019$.

\begin{figure}[t] 
\centering
\begin{minipage}[t]{0.40\linewidth}
       \flushleft
        \input Figures/figure10.tex
        \caption{Approximate locations of the Fisher zeros of $A_n^+(a)$ 
for $n\in\{25,50,75,100\}$.  The zeros scatter in the $a$-plane,
outlining a disk centred at the origin which is free from zeros.  There are
no zeros on the positive real axis, but an edge singularity is clearly being
formed at  $a_c^+$. A positive fraction of the zeros appear to accumulate 
on a circle of radius $a_c^+$.}
\label{figure-1}  
    \end{minipage}%
    \begin{minipage}[t]{0.40\linewidth}
        \flushright
        \input Figures/figure11.tex
        \caption{Approximate locations of the Fisher zeros of $A_n^+(a)$ 
for $n\in\{25,50,\ldots,175,200\}$.  These data include all the data in
figure \ref{figure-1}.  Notice that the zeros appear to pinch the
positive real axis to create an edge singularity
at the critical point $a_c^+$.  A positive fraction of the zeros appear to accumulate 
on a circle of radius $a_c^+$.}
\label{figure-3}  
    \end{minipage}%
\end{figure}

In general, the zeros of $G_N(a,t)$
are accumulating on a circle of radius $\lambda_a$ in the $t$ plane, while the zeros of
$g_N(a,t)$ are accumulating on the unit circle, also in the $t$-plane. 
The positions of the zeros are also close to $N$ vertices of a regular 
$(N\plus1)$-gon with vertices on the circle, and one vertex on the 
positive real axis which does not have a zero nearby.  
For even $N$ the zeros are found in complex conjugate pairs, and 
for odd $N$ there are $N\minus 1$ zeros in complex conjugate pairs, 
and one additional zero located on the negative real axis.  Increasing
$N$ (as a time variable) creates the illusion that new zeros are generated
on the negative real axis, and the existing zeros in the upper half plane move
clockwise along the circle (and in the lower half plane move anti-clockwise along the 
circle), while maintaining a distribution close to the vertices of a regular
$N$-gon on the unit circle.  This increases the density of zeros on the circle in a way
which becomes uniform in the limit $N\to\infty$. 
By theorem \ref{thm3.5} the zeros approach a uniform distribution on the circle
in the limit as $N\to\infty$.  This circle should become a natural boundary of the
generating function $G(a,t)$ (see equation \Ref{eqn9}) in the $t$-plane.

\begin{figure}[t]
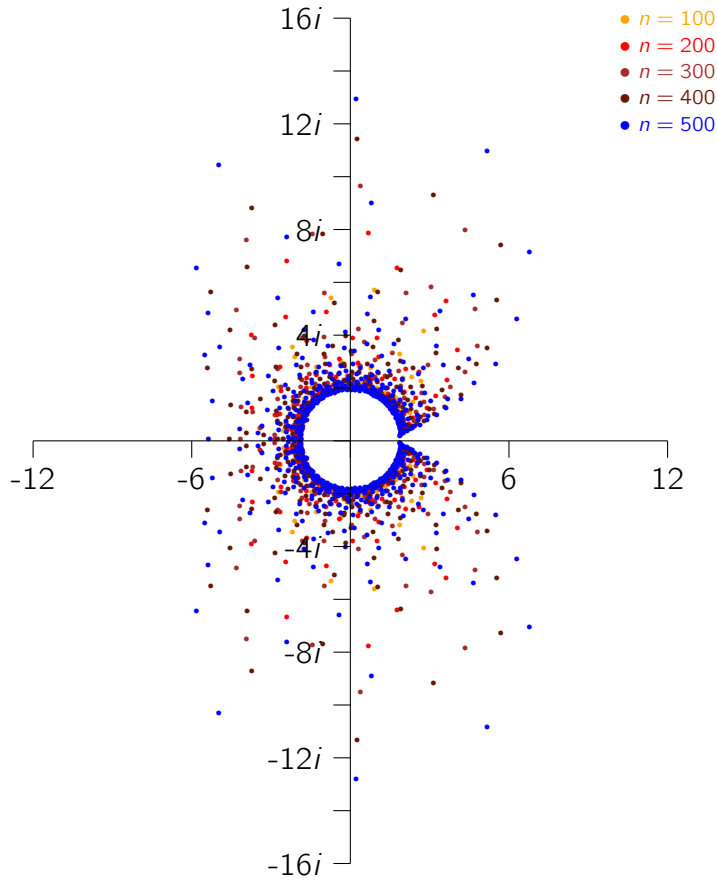

\input Figures/figure12.tex
\caption{Partition functions zeros of  the partition function $A_n^+(a)$  of two dimensional
adsorbing walks in the $a$-plane.  The zeros appear to be distributed outside a circle of radius
$a_c^+$ (the critical circle) and inside a vertical strip of width $w$ which
seems to increase with $n$.}
\label{figure-2}  
\end{figure}

\subsection{Fisher or partition function zeros}
\label{Fisher2A}   

The partition function $A_n^+(a)$ (see equation \Ref{eqn4}) was approximated
using the approximate enumeration data obtained by the GAS algorithm 
for adsorbing walks in $\mathL^2_+$ \cite{JvR16}.  The locations of zeros of $A_n^+(a)$ were
obtained to high accuracy by using Maple \cite{Maple}.  In figure \ref{figure-1} the
zeros are shown for small values $n$, namely $n\in\{25,50,75,100\}$,
and in figure \ref{figure-3} for values of $n$ in steps of $25$ to $n=200$.
These data suggest that the zeros form a characteristic pattern in the $a$-plane, 
such that a fraction is accumulating on a circle (called the \textit{critical circle}),
while the remaining zeros distribute
in the $a$-plane outside the the critical circle.  The results for larger values
of $n$, namely $n\in\{100,200,300,400,500\}$, are plotted in figure \ref{figure-2}.

By theorem \ref{thm2.6} there exists a $\rho\in (0,1)$ such that a positive fraction 
of zeros will be in an annulus centred at the origin and of inner radius 
$r=|a| = 1\minus \rho$ and outer radius $r=|a|= \sfrac{1}{1\minus\rho}$. 

The zeros plotted in figure \ref{figure-2}  appear to accumulate on a circle in the $a$-plane
with increasing $n$. However, many zeros are also found away from this circle,
and they do in fact appear to drift away from the origin with increasing $n$.  These
observations are not inconsistent with theorem \ref{thm2.6}.  In fact, it appears
that many zeros are restricted to an annulus containing the circle
and their angular distribution seems quite uniform (consistent with 
theorem \ref{cor222}), except in the
vicinity of the positive real axis, where some zeros are converging (with
increasing $n$) to the critical adsorption point $a_c^+$.  

The finite size \textit{excess free energy} $F^e_N(a)$ can be computed 
using equation \Ref{eqn19f}:
\begin{equation}
F_n^e(a) = F_n(a) - \log c_n^{(d-1)} = \sum_{k=1}^n \log (a\minus a_k) .
\label{eqn51}  
\end{equation}
The finite size \textit{energy} $E_n(a)$ per unit length and \textit{specific heat} 
$\C{C}_n(a)$ are the first and second derivatives of $F_n^e(a)$ to $\log a$, 
divided by length $n$:
\begin{equation}
\hspace{-2cm}
E_n(a) = \sfrac{a}{n}\sfrac{d}{da}\, F_n^e(a) 
= \Sfrac{a}{n} \sum_{k=1}^n \Sfrac{1}{a\minus a_k}, \q\hbox{and}\;\;
\C{C}_n(a) = a\sfrac{d}{da}\, E_n(a) 
= -\Sfrac{a}{n} \sum_{k=1}^n\Sfrac{a_k}{(a\minus a_k)^2}. 
\label{eqn52}  
\end{equation}

\begin{figure}[t]
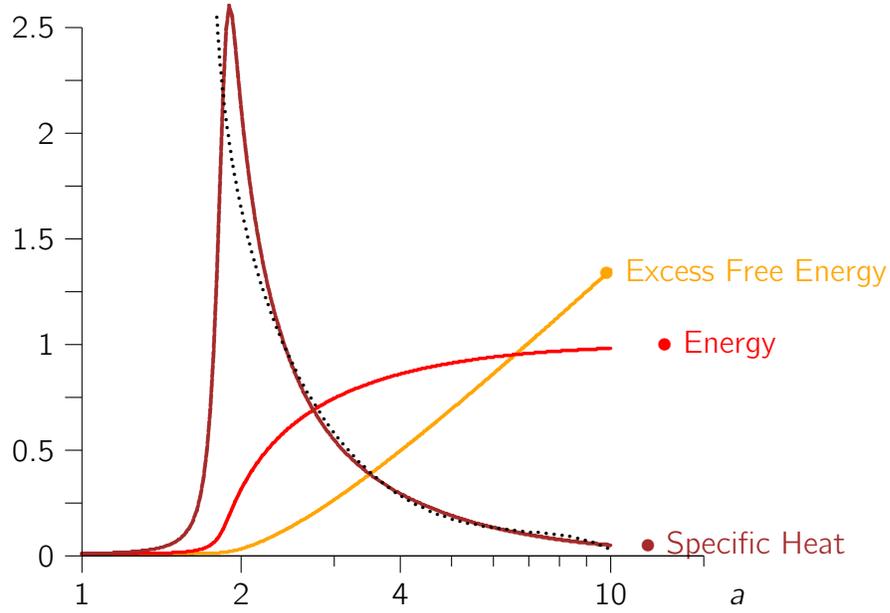

\input Figures/figure13.tex
\caption{The excess free energy $F^e_n(a)$ in the $a$-plane for two
dimensional adsorbing walks of length $n=500$ (see equation \Ref{eqn51}), 
plotted as a function 
of $a$ on a semilogarithmic graph.  The function is convex in $\log a$.
The energy per unit length, and the specific heat, are plotted on the
same graph, and were computed as shown in equation \Ref{eqn52}.  
The specific heat shows a sharp peak and is approximated by the dotted
curve (see equation \Ref{eqn54}).}
\label{figure-4}  
\end{figure}

The specific heat $\C{C}_n(a)$ has limiting behaviour $\lim_{n\to\infty} \C{C}_n(a)
\sim (a\minus a_c^+)^{-\alpha}$, where $2\minus \alpha = \sfrac{1}{\phi}$, and $\alpha$ 
is the specific heat exponent, while $\phi$ is the crossover exponent of adsorbing walks.  
It is thought that $\phi=\sfrac{1}{2}$ in all dimensions (numerical results in 
reference \cite{JvRR04} gives $\phi=0.5005\pm0.0036$), so $\alpha$ is close to or
equal to $0$.    That is, $\C{C}_n(a_c^+) \to \hbox{constant}$ as $n\to\infty$,
and since $\C{A}^+(a)$ is a constant when $a< a_c^+$ (see equation \Ref{eqn6}), one may consider
\begin{equation}
\lim_{n\to\infty} \C{C}_n(a)  \cases{
= 0, & \hbox{\norf if $a<a_c^+$}; \cr
\simeq A_0 - B_0 (a\minus a_c^+)^\phi, & \hbox{\norf if $a>a_c^+$},
}
\label{eqn51c}   
\end{equation}
provided that the activity $a$ is not much larger than $a_c^+$.  A least squares fit with
$\phi=\sfrac{1}{2}$ and the model
\begin{equation}
\C{C}_n(a) = A_0 - B_0 \sqrt{a\minus a_c^+} - C_0 (a\minus a_c^+) - D_0(a\minus a_c^+)^{1.5}
\label{eqn54}  
\end{equation}
to data for $n=500$ gives the dotted curve in figure \ref{figure-4} 
(where $A_0=3.096$, $B_0=-3.785$, $C_0=1.626$ and $D_0=-0.237$), 
approximating the shape of the specific heat curve well as $a$ increases from $a_c^+$.

\begin{figure}[t]
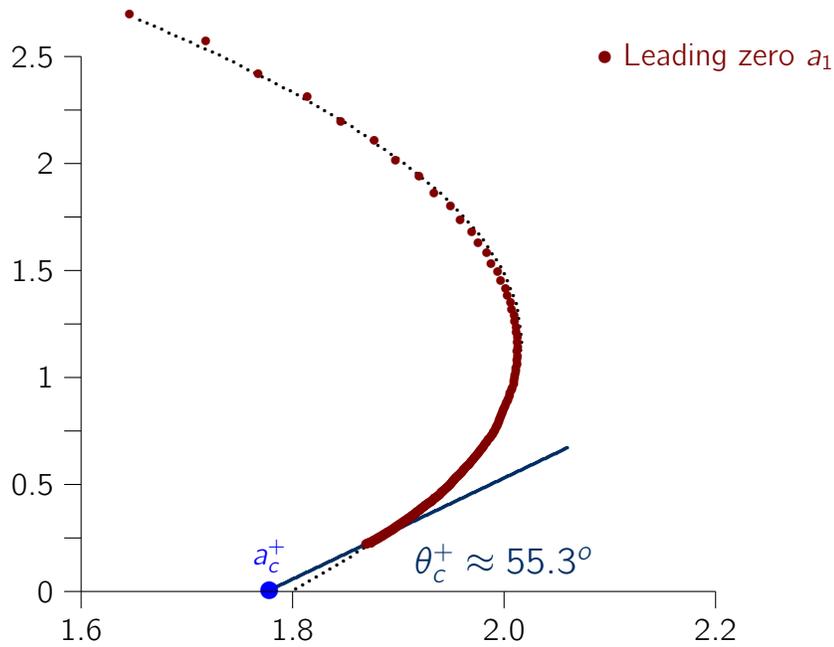

\input Figures/figure14.tex
\caption{The trajectory of the leading Fisher zero $a_1$ as a function of $n$ 
for adsorbing walks in the square lattice.  With increasing $n$ the leading zero approaches the
real axis, where it, and its complex conjugate, squeeze an edge singularity
in the limiting free energy of the model at the critical adsorption point
$a_c^+$.  The dotted line is a least squares parabolic
curve fitted to the trajectory of the leading zero.  Assuming a critical point at
$a_c^+=1.77564$ \cite{BGJ12} allows one to estimate the angle of incidence
$\theta^+_c \approx 55.3^o$.   The approximating parabola
intersects the real axis at the estimated critical point $1.8049$ at a critical angle
of size $58.2^o$.  If the least squares parabola is constrained to pass through the
estimated critical point at $1.77564$, then the critical angle has approximate size $53.7^o$.}
\label{figure-5}  
\end{figure}

Equation \Ref{eqn40A} shows that $E_n(a)$ approaches zero as $n\to\infty$ and
$a<a_c^+$.  Consequently, $\C{C}_n(a) \to 0$ as $n\to\infty$ and $a<a_c^+$.
Similarly, equation \Ref{eqn40A} shows that $E_n(a)$ approaches a constant
as $n\to\infty$ and $a\to\infty$.  That is, $\C{C}_n(a) \to \C{C}(a)\geq 0$ if 
$a>a_c^+$ and $n\to\infty$ (where $\C{C}(a)$ is
the limiting specific heat).  It is therefore reasonable to expect that 
$\C{C}_n(a_c^+) = \C{C}_m(a_c^+)$ at the adsorption critical point $a_c^+$ for
two different values of $n$ and $m$.  Solving for $a_c^+$ from the
equation $\C{C}_n(a) = \C{C}_m(a)$ for $m=2n$, and
$n\in\{200,210,220,230,240,250\}$,
gives estimates which are stable to three significant digits.  Taking the average gives
the estimate
\begin{equation}
a_c^+  = 1.765
, \q\hbox{in the square lattice},
\end{equation}
where the last digit is uncertain.  This compares well with the result from series enumeration
$a_c^+ = 1.77564$ \cite{BGJ12}.  The curves in figure \ref{figure-4} are determined from
the estimated locations of partition function zeros, from which the estimate for $a_c^+$ above
is determined.  This compares well with the result in reference \cite{JvR16}, namely
$a_c^+ = 1.779\pm 0.003$, which was determined using approximate enumeration data
generated with the GAS algorithm and analysed using conventional statistical methods.

The \textit{leading Fisher zero}, denoted
by $a_1$, is that zero of $A_n^+(a)$ with smallest principal argument in the
first quadrant in the $a$-plane.  The leading zero $a_1$ is a function of $n$, and it,
and its conjugate $\o{a}_1$, converge to
the edge singularity at $a_c^+$ on the positive real axis as $n\to\infty$.  Thus, as
$n$ increases, then $\Im a_1 \to 0$ and $\Re a_1 \to a_c^+$.  The rate at which $a_1$
moves to the real axis should be controlled by the crossover exponent $\phi$
of adsorbing walks. Since it is thought to be the case that $\phi=\sfrac{1}{2}$
in all dimensions $d\geq 2$ \cite{BWO99}, one may expect that
$\Im a_1 \sim \sfrac{c_i}{\sqrt{n}}$ and $\Re a_1 \sim a_c^+ + \sfrac{c_r}{\sqrt{n}}$.

The trajectory of $a_1$ in the first quadrant as $n$ increases
is shown in figure \ref{figure-5} for adsorbing walks in $\mathL^2_+$, and for
$n\in[10,500]$.  The estimate $a_c^+=1.77564$ is denoted by a bullet on the real axis
(see reference \cite{JvR16}) .  With increasing $n$ the trajectory of $a_1$ tends
to straigthen out and it approaches the real axis at a critical angle $\theta_c^+$ 
of about $55.3^o$. 

\begin{figure}[t]
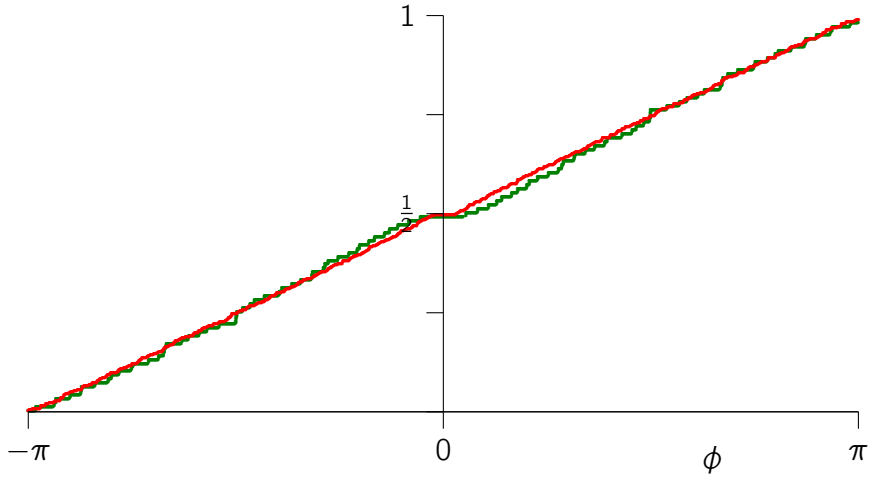

\input Figures/figure15.tex
\caption{The cumulative angular distribution function 
$\frac{1}{n} \alpha_n(-\pi,\phi)$ of Fisher zeros
(displayed for $n=400$ and $n=500$; see equation \Ref{eqn57}), 
plotted as a function of $\phi$, for adsorbing walks in the square lattice.  The curves
are smoother for larger values of $n$, and by equation \Ref{eqn34}, approach
a limiting curve bounded by $\lim_{n\to\infty}
\frac{1}{n} \alpha_n(-\pi,\phi) \leq \min\{1,\frac{1}{\pi}\phi + 1\}$.  The data
suggest that $\lim_{n\to\infty}
\frac{1}{n} \alpha_n(-\pi,\phi) = \frac{1}{2\pi}\phi + \frac{1}{2}$.}
\label{figure-6}  
\end{figure}

The leading zeros appear to accumulate on a curve which may be approximated by a 
quadratic curve in the $a$-plane.  If it is assumed that 
$\Re a_1 = a_c^+ + \alpha_0\, \Im a_1 + \alpha_1 (\Im a_1)^2$, then a
least squares fit to the data in figure \ref{figure-5} gives
\begin{equation}
\Re a_1 = 1.8049 + 0.6003\, \Im a_1 -0.4298 (\Im a_1)^2.
\end{equation}
This gives the estimate $a_c^+ \approx 1.8009$ and angle of incidence
$\theta_c^+ = 90^o - \arctan 0.6190 \approx 58.2^o$ (slightly larger than the $55.3^o$
calculated above).  It seems that the trajectory is fitted well by a quadratic curve,
but the values of the critical point $a_c^+$ and the critical angle $\theta_c^+$ 
obtained in this way are slightly larger than the estimates obtained earlier,
but not numerically inconsistent with those values.

Repeating the analysis above, but now with the critical value fixed at $a_c^+=1.77564$,
gives the least squares parabola
\begin{equation}
\Re a_1 = 1.77564 + 0.7347\, \Im a_1 -0.5220 (\Im a_1)^2.
\end{equation}
This gives the angle of incidence 
$\theta = 90^o - \arctan 0.7349 \approx 53.7^o$ (an angle slightly smaller than
calculated above).

\begin{figure}[t]
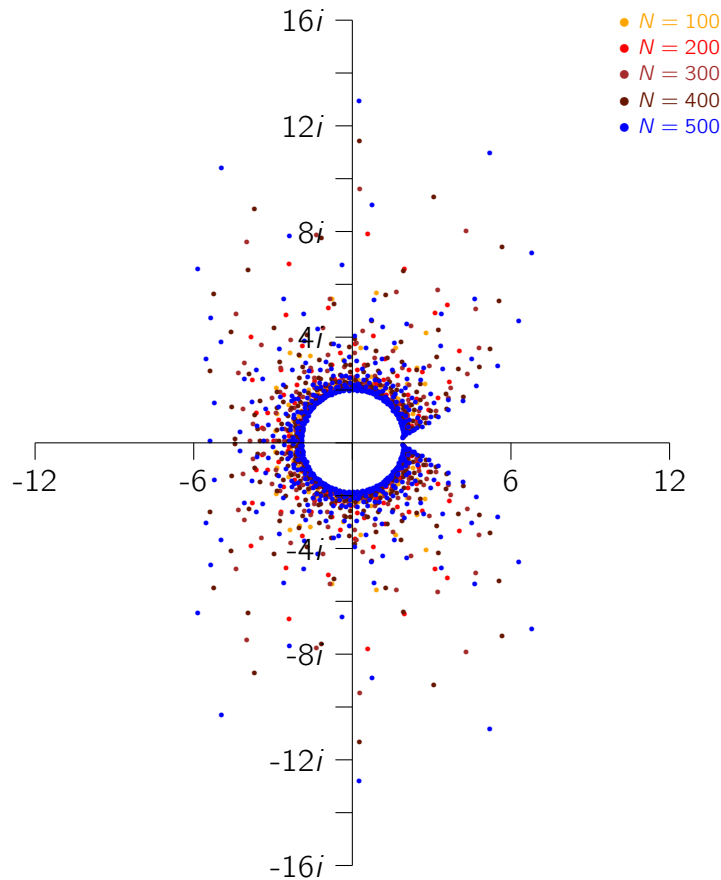

\input Figures/figure16.tex
\caption{Zeros of $G_N(a,t)$ (Yang-Lee zeros) in the $a$-plane for 
$t=\frac{1}{\mu_2}$.  The general distribution of the zeros is similar to
that seen for Fisher zeros (see figure \ref{figure-2}).  There are also 
little dependence on the value of $t$.  With increasing $N$ the
leading zero and its conjugate approaches the positive real axis
at $a_c^+$, creating an edge singularity in the limit as $N\to\infty$.}
\label{figure-7}  
\end{figure}

The \textit{cumulative angular distribution function}
$\hbox{\Large$\alpha$}_n(\phi)$ of the Fisher zeros of $A_n^+(a)$
around the origin in the $a$-plane is given by
\begin{equation}
\hbox{\Large$\alpha$}_n(\phi) = \sfrac{1}{n}\, \alpha_n(-\pi,\phi) 
\label{eqn57}  
\end{equation}
for adsorbing walks of length $n$ (and where $\alpha_n(\theta,\phi)$ is defined
in equation \Ref{eqn12}).  This is the fraction of Fisher zeros with
(complex) principal argument in the interval $(-\pi,\phi]$, and it increments
by $\sfrac{1}{n}$ each time a new zero enters the interval $(-\pi,\phi]$
with increasing $\phi$.   

The cumulative angular distibution function is plotted in figure \ref{figure-6}
for adsorbing walks in $\mathL^2_+$ for $n=400$ and $n=500$.  The function
increases steadily, except near $\phi=0$ where it has a constant part
which steadily decreases in length as $n$ increases (this corresponds to the
gap between leading Fisher zeros and the positive real axis near $\phi=0$).

\subsection{Yang-Lee zeros ($a$-plane zeros of $G_N(a,t)$)}
\label{sectionYangLee}  

The zeros of the partial generating function $G_N(a,t)$ (equation \Ref{eqn11})
in the $a$-plane are closely related to Fisher (partition function) zeros.  To 
distinguish these zeros from Lee-Yang zeros  (which are $t$-plane zeros of
$G_N(a,t)$), they will be called \textit{Yang-Lee zeros}.

The analysis of Yang-Lee zeros proceeds similar to Fisher zeros (as in 
section \Ref{Fisher2A}).  The zeros were determined by deflating
$G_N(a,t)$ for $t=\sfrac{1}{\mu_2}$  and are plotted in
figure \ref{figure-7} for $N\in\{100,200,300,400,500\}$.  In general, the 
location of zeros did not change in significant ways with changes in
the value of $t$.  The results for these zeros are very similar to those results seen
for Fisher zeros.  For example, the trajectory of the leading zero is
very similar to figure \ref{figure-5}, and the cumulative angular distribution of
zeros has the same profile as seen in figure \ref{figure-6}.

\begin{figure}[t]
\input Figures/figure17.tex
\caption{The angular distribution of partition function zeros in the square
lattice estimated by plotting
$\frac{50}{2\pi n}\, \alpha_n(\frac{(2m\minus 1)\pi}{100},
\frac{(2m\plus 1)\pi}{100})$ against $\frac{2m\pi}{50}$ 
for $m=-25,-24,-23,\ldots,24,25$,
and for $n\in\{100,200,\ldots,500\}$.  The darker colours correspond
to larger values of $n$, and the data points are interpolated by
line segments.  The data become less noisy with increasing $n$
and there is a notable decrease (to zero) of the
density of zeros close to the real axis (where $\phi=0$).}
\label{figure-8}  
\end{figure}

\subsection{The density profile and distribution of Fisher and Yang-Lee zeros}
\label{density2}   

Partition function zeros in figure \ref{figure-2}, or Yang-Lee zeros in
figure \ref{figure-7}, have angular and radial distribution profiles
around the origin in the $a$-plane.  The angular distribution of
Fisher zeros was already examined in figure \ref{figure-6},
and the same approach to Yang-Lee zeros gives a result which
is similar to this.

The distribution function $\alpha_n(\theta,\phi)$ is given by
\begin{equation}
\alpha_n(\theta,\phi) = n\,\hbox{\Large$\alpha$}_n(\phi) - 
n\,\hbox{\Large$\alpha$}_n(\theta), 
\end{equation}
and this may be estimated from the data in figure \ref{figure-6}.  This gives a
nearly uniform distribution, except for a small gap where the distribution
decreases for angles near $0$.  This is consistent with the
suggestion of theorem \ref{thm2.7a} on the angular distribution of
Fisher zeros.

An alternative approach to the angular distribution function 
is to count zeros in wedges with vertices at the origin (and where the
vertices of the wedges have size $\sfrac{2\pi}{M}$ in the $a$-plane).
If $M=50$ then the $a$-plane is partitioned into $50$ wedges
and the zeros are binned and counted in each wedge.  
This gives results as shown in figure \ref{figure-8}
for $n\in\{100,200,\ldots,500\}$. The darker colours correspond
to larger values of $n$, and the figure is obtained by  plotting
$\sfrac{50}{2\pi n} \, \alpha_n(\sfrac{(2m\minus 1)\pi}{100},
\sfrac{(2m\plus 1)\pi}{100})$ against $\sfrac{2m\pi}{50}$ for $m=-25,-24,\ldots,25$,
with the points interpolated by line segments.  The graph is very
uneven for small $n$, but this uneveness decreases as $n$ increases.  Also notable
is the low density in a (narrow) wedge containing the positive real axis.

Plotting Yang-Lee zeros instead of Fisher zeros gives very similar looking
results.

\begin{figure}[t]
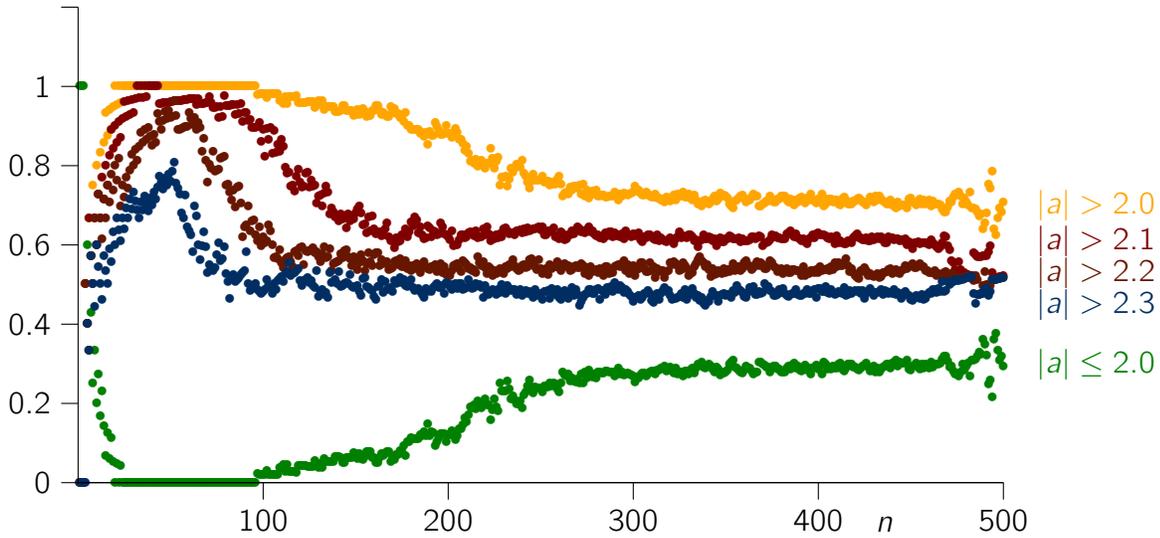

\input Figures/figure18.tex
\caption{The fraction of partition function zeros (zeros of $A_n^+(a)$ in the
$a$-plane) of adsorbing walks in the square lattice.  The data displayed
show the fraction of zeros at distances closer than $2$ (green points along the 
lowest curve), and then farther away than $2$,$2.1$, $2.2$ and $2.3$ from 
the origin in the $a$-plane.  Determining the same data for Yang-Lee zeros
(zeros of $G_N(a,t)$ with $t=\frac{1}{\mu_2}$) gives a similar result, 
with data points which largely coincide with the data plotted here.  The data
appear to flatten to constant values with increasing $n$.  This suggests that
the number density of zeros decreases inversely with distance from
the origin if $|a|$ is large and in the limit that $n\to\infty$.}
\label{figure-9}  
\end{figure}

In the radial direction Fisher zeros appear to accumulate on the
critical circle $|a|=a_c^+$ in the $a$-plane.  The fraction
of zeros inside or outside a disk $|a|< r$ is a function of
$n$, and may be approaching limiting values as $n$ increases to
infinity.  This is tested in figure \ref{figure-9}.  The data points marked by
$|a|<2.0$ correspond to the fraction of zeros inside the disk
$|a|\leq 2$ (this is the same as the fraction of zeros inside the
annulus $a_c^+ \leq |a| \leq 2$).  The fraction is small for $n<100$,
but increases steadily and appears to approach a constant fraction of
about $0.30$ for $n>300$.

A similar situation is seen for zeros farther away from the origin.
For example, the fraction of zeros farther away than a distance of $2.0$
from the origin (that is, the fraction of zeros outside a disk $|a|\leq 2$
at the origin), is shown by the data points accumulating
along the top curve of the graph in figure \ref{figure-9}.  For small
values of $n$ the fraction is nearly equal to $1$, but this decreases
and levels off at about $0.70$ for $n>300$.

The data outside a disk $|a|\leq r$ decreases with increasing $r$,
as seen for $r\in\{2.0,2.1,2.2,2.3\}$ in figure \ref{figure-9}.    In each case
the fraction of zeros levels off (within noise) to a constant,  independent
of $n$.  The results for Yang-Lee zeros (zeros of $G_N(a,t)$) with
$t=\frac{1}{\mu_2}$ are similar to the results for Fisher zeros shown
in figure \ref{figure-9}.

The results in figure \ref{figure-9} suggest that only a fraction of Fisher 
and Yang-Lee zeros accumulate on the critical circle $|a|=a_c^+$.  That is,
it appears that a constant fraction of zeros accumulate on the critical circle,
while the inside of the critical circle is empty, and the remaining zeros are
distributed outside the critical circle in the $a$-plane.
The radial density (integrated in the angular direction)
approaches a constant value with increasing $n$.  This suggests, in the limit
as $n\to\infty$, that there are zeros at arbitrary distances 
from the origin.  However, since the distribution of zeros is integrable,
the \textit{number density} of zeros should drop off faster than
the inverse squared distance from the origin.

\section{Numeral Results in the Cubic Lattice}
\label{Numerical3}   

\subsection{Lee-Yang zeros of $G_N(a,t)$ in the cubic lattice}
\label{Lee-Yang3} 

Similar to the results in the square lattice, the Lee-Yang zeros of $G_N(a,t)$
will accumulate on a circle of radius $\lambda_a$ (see equation \Ref{eqn13})
as $N\to\infty$ in the $t$-plane.  By scaling $t$ with $\lambda_a$ (see equation \Ref{eqn15}),
the Lee-Yang zeros will accumulate on the unit circle in the
$t$-plane (by theorem \ref{thm3.4}), and their distribution will be (asymptotically)
uniform in the limit as $N\to\infty$. 

The Lee-Yang zeros will be examined by considering the partial sum $g_N(a,t)$
in equation \Ref{eqn15} for data obtained for $\mathL^3_+$.  The estimated
locations of these zeros are shown in figures \ref{figure3-t2-1} and \ref{figure3-t2-3},
for $N=10$ and $a=1$, and for $a\in\{1,a_c^+,3\}$, respectively, where $a_c^+$ is
the critical adsorption activity estimated in reference \cite{JvR16} 
($a_c^+ \approx 1.306$).

\begin{figure}[t] 
\centering
\begin{minipage}[t]{0.40\linewidth}
       \flushleft
        \input Figures/figure19.tex
        \caption{Approximate locations of the zeros of $g_N(a,t)$ for $N=10$
and $a=1$ in the cubic lattice.  The zeros are located approximately on the 
unit circle, and close to 10 vertices of a regular 11-gon with centre at the origin
(oriented with one vertex at the point $(1,0)$).  The rays from the origin
are the spokes of the regular 11-gon.  
}
\label{figure3-t2-1}  
    \end{minipage}%
    \begin{minipage}[t]{0.40\linewidth}
        \flushright
        \input Figures/figure20.tex
        \caption{Approximate locations of the zeros of $g_N(a,t)$ for $N=10$
and for $a=1$, $a=1.306\approx a_c^+$,
and $a=3$ in the cubic lattice.  The zeros for $a>1$ have moved slightly away from the
unit circle, but are still fairly close to it. 
}
\label{figure3-t2-3}  
    \end{minipage}%
\end{figure}

The best least-squares ellipse through the zeros in figure
\ref{figure3-t2-1} has centre at $(0.012,0)$, half short axis of length $1.031$ 
and half long axis of length $1.042$.  The rays from the origin in this
figure intersect the unit circle in the vertices of a regular $11$-gon
(with one vertex on the real axis).  Each vertex of the $11$-gon is close
to a Lee-Yang zero, except for the vertex on the real axis.

In figure \ref{figure3-t2-3} the dependence of the Lee-Yang zeros on the
value of $a$ is examined. The data plotted in this figure correspond to
$a\in\{1,a_c^+,3\}$ with $a_c^+$ approximated by $1.306$ \cite{JvR16}.
For both $a=1.306$, and $a=3$, the zeros have moved slightly
off the unit circle, but  they are still located on a least squares ellipsiod 
with centre close to the origin (in the case of $a=1.306$, 
with centre at $(0.0286,0)$, half short axis of approximate length $1.207$ 
and half long axis of approximate length $1.211$; 
and for $a=3$, with centre at $0.0487,0)$, half short axis of 
approximate length $1.209$ and half long axis of approximate length
$1.216$).

The data for $N=100$ are shown in figures \ref{figure3-t2-2} and
\ref{figure3-t2-4}.  The zeros are evenly distributed along the unit
circle for $a=1$ in figure \ref{figure3-t2-2}, close to the vertices of
a regular $101$-gon with centre at the origin and one vertex on the positive
real axis at the point $(1,0)$.  The best least-squares ellipse through these
zeros has centre at $(0.001,0)$, approximate half short axis of length
$1.001$ and approximate half long axis of length $1.011$.   Changes in the values of
$a$ do not have a significant effect on the location of the zeros of $g_N(a,t)$, since these
are are already scaled by estimates of $\lambda_a$.  Thus, the
zeros for $a\in\{1,a_c^+,3\}$ shown in figure \ref{figure3-t2-4} are all
close to the unit circle and are evenly distributed along it.  The best
least-squares ellipse through the zeros when $a=1.306$ in figure \ref{figure3-t2-4}
has centre $(0.003,0)$, half short axis of approximate length $1.026$ and half
long axis of approximate length $1.026$.  If $a=3$, then the least squares
ellipse has centre at $(0.006,0)$, half short axis of approximate length $1.025$ and
half long axis of approximate length $1.027$.

\begin{figure}[t] 
\centering
\begin{minipage}[t]{0.40\linewidth}
       \flushleft
        \input Figures/figure21.tex
        \caption{Approximate locations of the zeros of $g_N(a,t)$ for $N=100$
and $a=1$ in the cubic lattice.  The zeros are approximately on the unit circle, 
and close to 100 vertices of a regular 101-gon with centre at the origin, and
oriented with one vertex at the point $(1,0)$.  
}
\label{figure3-t2-2}  
    \end{minipage}%
    \begin{minipage}[t]{0.40\linewidth}
        \flushright
        \input Figures/figure22.tex
        \caption{Approximate locations of the zeros of $g_N(a,t)$ for $N=100$
and for $a=1$, $a=1.306\approx a_c^+$, and $a=3$ in the cubic lattice.  
All the zeros are close to the unit circle. As in figure \ref{figure3-t2-2}
the rays are the spokes of a regular $101$-gon.  
}
\label{figure3-t2-4}  
    \end{minipage}%
\end{figure}

\subsection{Fisher or partition function zeros}
\label{Fisher3}  

The partition function $A_n^+(a)$ of adsorbing walks in $\mathL^3_+$ was
approximated using the approximate microcanonical data obtained 
from the GAS algorithm simulation for adsorbing walks in reference \cite{JvR16}.
The locations of zeros in the complex plane were similarly determined 
by deflating the polynomials using Maple \cite{Maple}.  The results for $n\in\{25,50,75,100\}$
are plotted in figure \ref{figure3-1}, and for $n\in\{100,200,\ldots,500\}$ are shown
in figure \ref{figure3-2}.

By theorem \ref{thm2.6} a positive fraction of partition function zeros will be 
located in an annulus $\{a\svv \hbox{$1\minus\rho \leq r \leq \sfrac{1}{1\minus\rho}$} \}$
for some value of $\rho\in(0,1)$. 

The zeros plotted in figure \ref{figure3-2}  appear to accumulate on a circle 
in the $a$-plane with increasing $n$. However, many zeros are also distributed
away from this circle, and with increasing $n$ zeros are found at further
distances from the origin.  These observations are not inconsistent with 
theorem \ref{thm2.6}.  Many zeros are located in an annulus containing the circle of
critical radius $a_c^+$ (by theorem \ref{thm2.7}) and their angular distribution 
appears to be very even, except in the vicinity of the positive real axis, 
where the leading zeros are converging to the edge singularity (with
increasing $n$) at the critical adsorption point $a_c^+$.  

\begin{figure}[t] 
\centering
\begin{minipage}[t]{0.4\linewidth}
       \flushleft
        \input Figures/figure23.tex
        \caption{Partition function (Fisher) zeros of $A_n^+(a)$  in the $a$-plane for
adsorbing walks in the cubic lattice, for $n\in\{25,50,75,100\}$.}
\label{figure3-1}  
    \end{minipage}%
    \begin{minipage}[t]{0.4\linewidth}
        \flushright
        \input Figures/figure24.tex
        \caption{Partition function (Fisher) zeros of $A_n^+(a)$  in the $a$-plane for 
adsorbing walks in the cubic lattice, for $n\in\{100,200,300,400,500\}$.}
\label{figure3-2}  
    \end{minipage}%
\end{figure}

\begin{figure}[t]
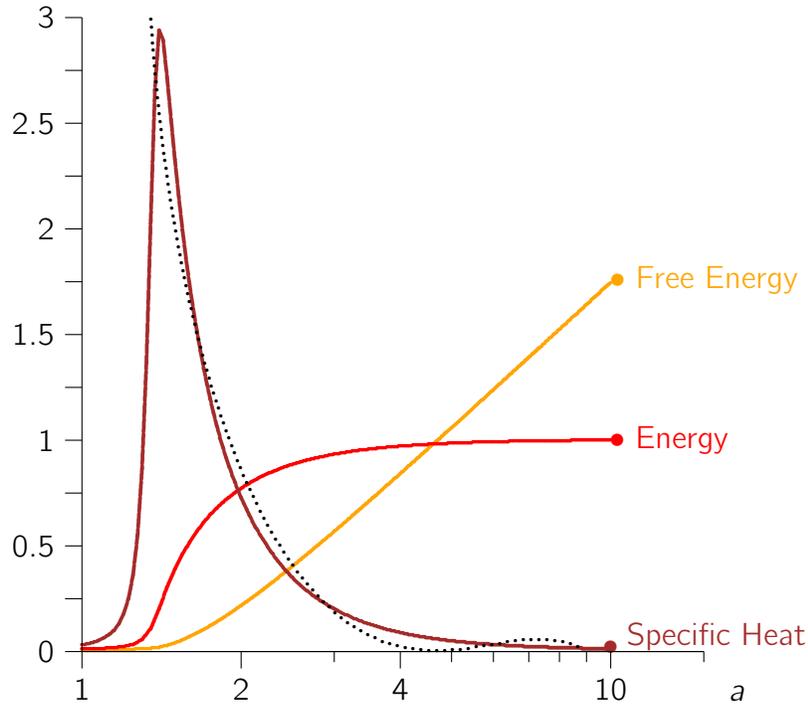

\centering
\input Figures/figure25.tex
\caption{The excess free energy $F^e_n(a)$ in the $a$-plane for three
dimensional adsorbing walks of length $n=500$ (see equation \Ref{eqn51}), plotted as a function 
of $a$ on a semilogarithmic graph.  The function is convex in $\log a$.
The energy per unit length, and the specific heat, are plotted on the
same graph, and were computed as shown in equation \Ref{eqn52}.  
The specific heat shows a sharp peak and is approximated by the dotted
curve (see equation \Ref{eqn54}).}
\label{figure3-4}  
\end{figure}

The finite size \textit{excess free energy} $F^e_n(a)$ can be computed 
using equation \Ref{eqn19f}:
\begin{equation}
F_n^e(a) = F_n(a) - \log c_n^{(d-1)} = \sum_{k=1}^n \log (a\minus a_k) .
\label{eqn51d}  
\end{equation}
The finite size \textit{energy} $E_n(a)$ per unit length and \textit{specific heat} 
$\C{C}_n(a)$ are the first and second derivatives of $F_n^e(a)$ to $\log a$, 
divided by length $n$ (see equation \Ref{eqn52}).   Assuming that the specific heat $\C{C}_n(a)$ 
has limiting behaviour as shown in equation \Ref{eqn51c} with $\phi=\sfrac{1}{2}$, 
then $\C{C}_n(a)$ should be given by equation \Ref{eqn52} when $a>a_c^+$ in the
limit as $n\to\infty$.  A least squares fit gives the dotted curve in 
figure \ref{figure3-4} (where $A_0=4.123$, $B_0=-5.987$, $C_0=2.852$
and $D_0=-0.443$), approximating the shape of the specific heat
curve well as $a$ increases from $a_c^+$.

Equation \Ref{eqn51c} shows that $\C{C}_n(a) \to 0$ as $n\to\infty$ and $a<a_c^+$, 
and $\C{C}_n(a) \to \C{C}(a) \geq 0$ if $a>a_c^+$ and $n\to\infty$.  It may
reasonably be expected that $\C{C}_n(a_c^+) = \C{C}_m(a_c^+)$ at the critical
point.  Solving this for $m=2n$ and for $n\in\{200,210,220,230,240,250\}$,
give estimates which are stable to three decimal places.  Taking the average gives
the estimate
\begin{equation}
a_c^+  = 1.326, \q\hbox{in the cubic lattice},
\end{equation}
where the last digit is uncertain.  This compares well with (but is slightly larger
than)  the result obtained
in reference \cite{JvR16}, namely $a_c^+ = 1.306(7)$ \cite{JvR16}.

\begin{figure}[t]
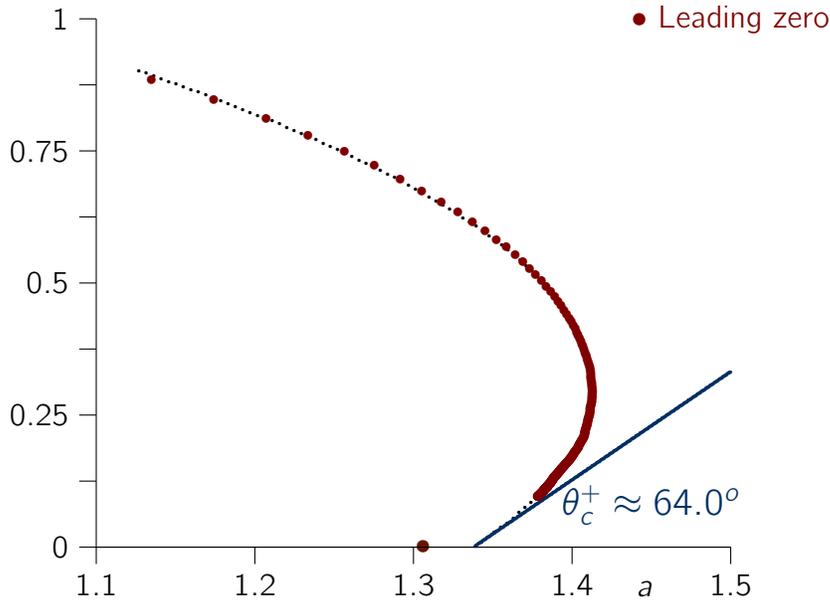

\input Figures/figure26.tex
\caption{The trajectory of the leading Fisher zero $a_1$ as a function of $n$ 
for adsorbing walks in the cubic.  With increasing $n$ the leading zero approaches the
real axis, where it (together with its complex conjugate) squeeze an edge singularity
in the limiting free energy of the model at the critical adsorption point
$a_c^+$ (denoted by $\bullet$ on the real axis).  The dotted line is a least squares quadratic
curve fitted to the trajectory of the leading zero.  It intersects the 
real axis at the critical point, estimated in this way to be approximately
equal to $1.34$.  The angle at which the curve intersects the real axis is
about $64^o$.  If the least squares parabola is constrained to pass through the
estimated critical point $a_c^+=1.306$, then its angle of intersection is approximately
$54.6^o$.}
\label{figure3-5}  
\end{figure}

The leading Fisher zeros $a_1$ are plotted in figure \ref{figure3-5}.
The estimate $a_c^+=1.306$ from reference \cite{JvR16} is denoted by the bullet 
on the real axis.  With increasing $n$ zeros appear to move along a 
curve and it approaches the real axis at an angle $\theta$ as $n \to\infty$.
A least squares fit to the curve, using the model  $\Re a_1 = 
a_c^+ + \alpha_0 \Im a_1 + \alpha_1 (\Im a_1)^2$ gives
\begin{equation}
\Re a_1 = 1.340 + 0.4868 \Im a_1 -0.8035 (\Im a_1)^2.
\end{equation}
This gives the estimate $a_c^+ \approx 1.340$ and angle of incidence
$\theta = 90^o - \arctan 0.4846 \approx 64.0^o$.  The location of the leading
zeros are well approximated by the quadratic curve,  and this may indicate that the
estimate in \cite{JvR16} slightly underestimates the location of the critical
point.

Repeating the analysis above, but now with the critical value fixed at $a_c^+=1.306$
(so that the least squares parabola passes through the point $a_c^+$ on the real
axis), gives
\begin{equation}
\Re a_1 = 1.306 + 0.7350 \Im a_1 -1.10578 (\Im a_1)^2.
\end{equation}
This gives the angle of incidence 
$\theta = 90^o - \arctan 0.7350 \approx 54.8^o$ (a smaller angle than
calculated above). 

The estimates for the critical angle above (namely $64.0^o$, and $54.8^o$)
show that there remains significant uncertainty in its value.

The \textit{cumulative angular distribution function} $\hbox{\Large$\alpha$}_n(\phi)$ 
of the partition function zeros of adsorbing walks of length $n$ 
around the origin in the $a$-plane is defined in equation \Ref{eqn57}
(where $\alpha_n(\theta,\phi)$ is defined in equation \Ref{eqn12}).  
That it, $\hbox{\Large$\alpha$}_n(\phi)$ is the fraction of Fisher zeros with
(complex) principal argument in the interval $(-\pi,\phi]$, and it increments in 
steps of size $\sfrac{1}{n}$ each time a new zero enters the interval $(-\pi,\phi]$
with increasing $\phi$.   $\hbox{\Large$\alpha$}_n(\phi)$ is plotted in
figure \ref{figure3-6} for adsorbing walks in $\mathL^3_+$ and for
$n=400$ and $n=500$.  For these finite values of $n$ it is a step-function,
with a wider constant part straddling the positive real axis 
where $\phi=0$.  The constant part steadily decreases in length as $n$ increases 
(this corresponds to the leading Fisher zero approaching the real
axis with increasing $n$).

The distribution function $\alpha_N(\theta,\phi)$ is given by
\begin{equation}
\alpha_n(\theta,\phi) = n\,\hbox{\Large$\alpha$}_n(\phi) - 
n\,\hbox{\Large$\alpha$}_n(\theta), 
\end{equation}
and this can be estimated from the data in figure \Ref{figure3-6}. 
Polynomial fits to $\hbox{\Large$\alpha$}_n(\theta)$ give a very
uniform profile for $\alpha_n(\theta,\phi)$,  consistent with 
theorem \ref{thm2.7a} on the angular distribution of Fisher zeros.

\subsection{Yang-Lee zeros ($a$-plane zeros of $G_N(a,t)$)}
\label{Yang-Lee3}   

The Yang-Lee zeros (see section \ref{sectionYangLee}) of the partial generating 
function $G_N(a,t)$ (equation \Ref{eqn11}) are closely related to Fisher (partition function) 
zeros.  These zeros were estimated using Maple \cite{Maple}, and are
plotted for $N\in\{100,200,300,400,500\}$ in figure \ref{figure3-7}.

The distribution of Yang-Lee zeros, and their other properties, are very 
similar to those of Fisher zeros; see section \ref{Fisher2A}.  The data in figure
\ref{figure3-7} are for $t=\sfrac{1}{\mu_3}$, and it is found that 
the locations of the zeros, and their distribution, are not very sentitive to changes
in $t$ (the data in figure \ref{figure3-7} are very similar to the
data shown in figure \ref{figure3-2}).  In addition, the position and trajectory 
of the leading zero is very similar to the data in figure \ref{figure3-5},  
while the cumulative angular distribution of zeros has the same profile 
as seen in figure \ref{figure3-6}.

\begin{figure}[t]
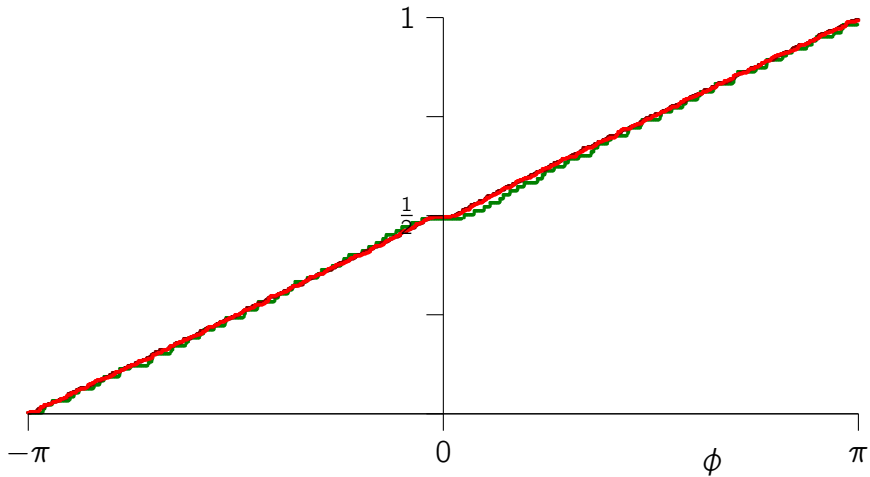

\input Figures/figure27.tex
\caption{The cumulative angular distribution function of Fisher zeros of
adsorbing walks in the cubic lattice, $\frac{1}{n}\, \alpha_n(-\pi,\phi)$
(displayed for $n=400$ and $n=500$ in the cubic lattice as
a function of $\theta$; see equation \Ref{eqn57}).  These curves are step functions
and they become smoother for larger values of $n$. By equation \Ref{eqn34} 
the limit curve is bounded by $\lim_{n\to\infty}
\frac{1}{n}\, \alpha_n(-\pi,\phi) \leq \min\{1,\frac{1}{\pi}\phi + 1\}$.  The data
suggest that $\lim_{n\to\infty}
\frac{1}{n}\, \alpha_n(-\pi,\phi) = \frac{1}{2\pi}\phi + \frac{1}{2} .$}
\label{figure3-6}  
\end{figure}

\subsection{The density profile and distribution of Fisher and Yang-Lee zeros}
\label{density3}   

The angular and radial distribution of partition function zeros in 
figure \ref{figure3-2}, or Yang-Lee zeros in figure \ref{figure3-7}, 
are similar.  The angular distribution of partition function zeros
was examined in figure \ref{figure3-6}, where the cumulative angular
distribution function is shown.  Yang-Lee zeros of $G_N(a,t)$ have a
cumulative distribution which is similar to figure \ref{figure3-6},
and it is not shown here.

The angular distribution of zeros may also be examined by counting
zeros in wedges around the origin.  If the plane is partitioned by $M$ wedges
at the origin in the $a$-plane with vertex angles of size $\sfrac{2\pi}{M}$, then 
the angular distribution can be estimated by counting zeros in each wedge.
Binning the partition function zeros for each wedge, and plotting the 
results against bin number gives figure \ref{figure3-8} (where $M=50$
and for $n\in\{100,200,\ldots,500\}$). The darker colours correspond
to larger values of $n$, and the figure is obtained by  plotting
$\sfrac{50}{2\pi n}\, \alpha_n(\sfrac{(2m\minus 1)\pi}{100} , 
\sfrac{(2m\plus 1)\pi}{100})$ against $\sfrac{2m\pi}{50}$ for $m=-25,-24,\ldots,25$,
with the points interpolated by line segments.  The graph is very
uneven for small $n$, but this uneveness subsides as $n$ increases.  Also notable
is the low density in a (narrow) wedge containing the positive real axis. 

Plotting Yang-Lee zeros instead of partition function zeros gives very similar looking
results.

\begin{figure}[t]
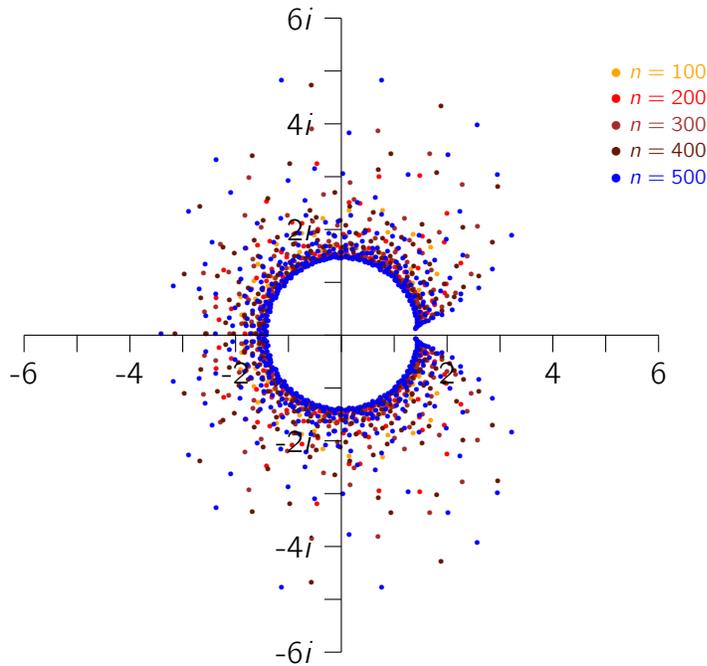

\input Figures/figure28.tex
\caption{Zeros of $G_N(a,t)$ (Yang-Lee zeros) for adsorbing walks
in the cubic lattice in the $a$-plane for $t=\frac{1}{\mu_3}$.  This distribution 
of the zeros is similar to that seen for Fisher zeros (see figure \ref{figure3-2}).  
Examination of the data show that this pattern of zeros is not very
sensitive to the value of $t$ in the generating function.  With increasing $N$ the
leading zero and its conjugate approaches the positive real axis
at the adsorption critical point $a_c^+$, creating an edge singularity 
in the limit as $N\to\infty$.}
\label{figure3-7}  
\end{figure}

In the radial direction the partition function zeros appear to accumulate on a
circle $|a|=a_c^+$ in the $a$-plane with increasing $n$.  The fraction
of zeros inside a disk $|a|< r$, or outside disks $|a|<r$, is a function of
$n$, but this may be approaching limiting values as $n$ increases to
infinity.  This possibility is examined in figure \ref{figure3-9}.  The data points marked by
$|a|\leq 1.45$ in this figure correspond to the fraction of zeros inside the disk
$|a|\leq 1.45$ (this is the same as the fraction of zeros inside the
annulus $a_c^+ \leq |a| \leq 1.45$).  The fraction is small for $n<150$,
but increases steadily and seems to approach a constant value of about
$0.30$ for $n>300$.

A similar situation is seen for zeros farther away from the origin.
For example, the fraction of zeros farther away than a distance of $1.45$
from the origin (that is, the fraction of zeros \textit{outside} a disk $|a|\leq 1.45$
at the origin), is shown by the data points accumulating
along the top curve of the graph in figure \ref{figure3-9}.  For small
values of $n$ the fraction is nearly equal to $1$, but this decreases
and levels off at a density of about $0.70$ for $n>300$.

The density outside disks $|a|\leq r$ decreases with increasing $r$,
as seen for $r\in\{1.45,1.55,1.65,1.75\}$ in figure \ref{figure3-9}.    In each case
the fraction of zeros levels off (within noise) to a constant,  independent
of $n$.  The results for Yang-Lee zeros (zeros of $G_N(a,t)$ with
$t=\frac{1}{\mu_2}$) are similar to the results for partition function zeros shown
in figure \ref{figure3-9}.

\begin{figure}[t]
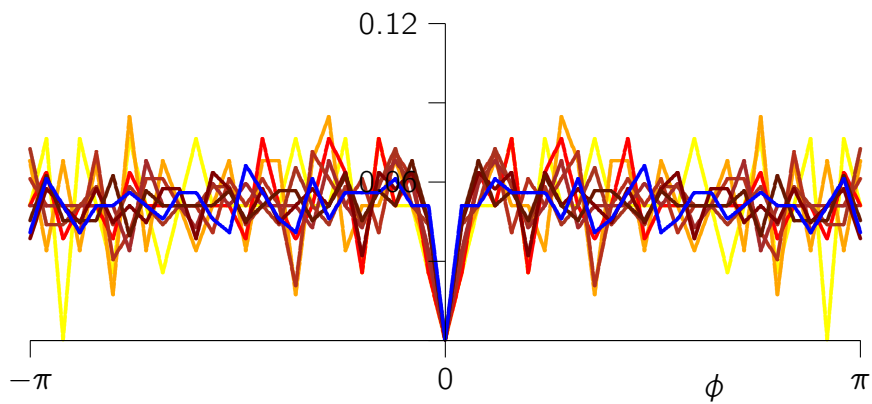

\input Figures/figure29.tex
\caption{The angular distribution of partition function zeros in the cubic
lattice estimated by plotting
$\frac{50}{2\pi n}\, \alpha_n(\frac{(2m - 1)\pi}{100} , 
\frac{(2m + 1)\pi}{100})$ against $\frac{2m\pi}{50}$ 
for $m=-25,-24,-23,-22,\ldots,24,25$,
and for $n\in\{100,200,\ldots,500\}$.  The darker colours correspond
to larger values of $n$, and the data points are interpolated by
line segments.  The data become less noisy with increasing
$n$, and there is a notable decrease (to zero) of the
density of zeros close to the positive real axis (where $\phi=0$).}
\label{figure3-8}  
\end{figure}

Similar to the observation for adsorbing square lattice walks,
the results in figure \ref{figure3-9} suggest that only a fraction of Fisher 
and Yang-Lee zeros accumulate on the critical circle $|a|=a_c^+$.  That is,
it appears that a constant fraction of zeros accumulate on the critical circle,
while the inside of the critical circle is empty.
The radial density (integrated in the angular direction)
approaches a constant value with increasing $n$.  This suggests, in the limit
as $n\to\infty$, that there are zeros at arbitrary distances 
from the origin.  However, since the distribution of zeros is integrable,
the \textit{number density} of zeros should drop off faster than
the inverse squared distance from the origin.

\section{Conclusions}
\label{conclusions}   

In this paper the properties of partition and generating functions zeros
of models of adsorbing walks in the square and cubic lattices were examined
numerically.  
In addition, it was shown that theorems on the distribution of polynomial zeros
can be used to prove theorems on the distribution of partition and generating
function zeros in these models.  This includes theorem \ref{thm3.4}, which shows that
generating function zeros in the $t$-plane converge on a circle -- these are
the Lee-Yang zeros, and by theorem \ref{thm3.5} they converge to a uniform 
distribution in the limit as the length of the walk approaches infinity.  These
results are corollaries of the theorem of Hughes and Nikechbali
(theorem \ref{thm3.3}), and the theorems of
Erd\"os and Tur\'an (theorem \ref{thm3.6}) and Erd\'elyi (theorem \ref{thm3.7}).

Numerical results on the Lee-Yang zeros in the $t$-plane were given in sections
\ref{Lee-Yang2} and \ref{Lee-Yang3}.  The results show that the zeros are
distributed evenly on a circle of radius $\lambda_a^{-1}$, even for small values
of $N$ as seen in figures \ref{figure-t2-1}, \ref{figure-t2-3}, \ref{figure-t2-2} 
and \ref{figure-t2-4} for square lattice adsorbing walks, and
figures \ref{figure3-t2-1}, \ref{figure3-t2-3}, \ref{figure3-t2-2} 
and \ref{figure3-t2-4} for cubic lattice adsorbing walks.  Moverover, the
zeros of the generating function $G_N(a,t)$ were seen to be close to 
$N$ vertices of a regular $(N\plus 1)$-gon.  As $N$ increases this gives
an increasingly uniform distribution of the zeros on the circle.

\begin{figure}[t]
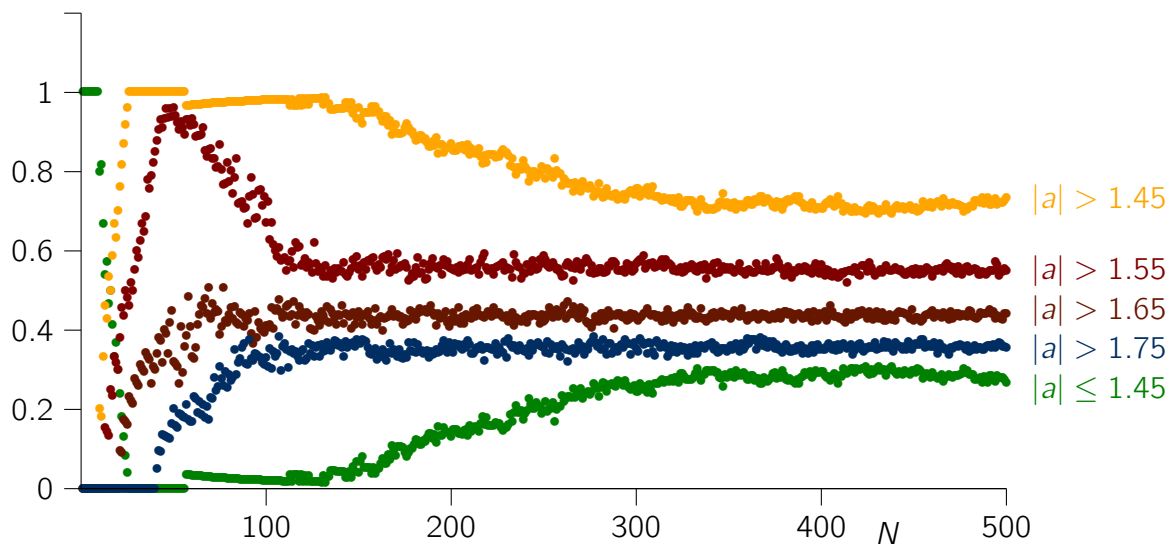

\input Figures/figure30.tex
\caption{The fraction of partition function zeros (zeros of $A_n^+(a)$ in the
$a$-plane) of adsorbing walks in the cubic lattice.  The data displayed
show the fraction of zeros at distances closer than $1.45$ (green points along the 
lowest curve), and then farther away than $1.45$, $1.55$, $1.65$ and $1.75$ from 
the origin in the $a$-plane.  Determining the same data for Yang-Lee zeros
(zeros of $G_N(a,t)$ with $t=\frac{1}{\mu_3}$) gives similar results, 
with data points which largely coincide with the data plotted here.  The data
appear to flatten out to constant values with increasing $n$.  This implies that
the number density of zeros decreases inversely with distance from
the origin if $|a|$ is large and in the limit that $n\to\infty$.}
\label{figure3-9}  
\end{figure}

The distribution of partition function (or Fisher) zeros was also examined.
These zeros were shown to be distributed with positive density in an annulus in the
complex $a$-plane in theorem \ref{thm2.6}.  
The angular distribution of Fisher zeros was also considered by examining 
the function $\alpha_n(\theta,\phi)$.  This, in particular, gives, in the
limit as $n\to\infty$, an angular distribution $q(\theta)$, as defined 
in equation \Ref{eqn22}.  Bounds on $q(\theta)$ are given in theorem
\ref{thm2.7a}.

Numerical results on the distribution of Fisher zeros are consistent with
the theorems in section \ref{Fisher2}.  The distribution of Fisher zeros is
typically similar to the data shown in figures \ref{figure-1}, \ref{figure-3} and \ref{figure-2}
for square lattice adsorbing walks, and figures \ref{figure3-1} and \ref{figure3-2}
for cubic lattice adsorbing walks.  It was shown that thermodynamic quantities
(the free energy, energy and specific heat) can be computed from the
Fisher zeros (see figures \ref{figure-4} and \ref{figure3-4}) and that
these are consistent with results obtained by other means; see for example
reference \cite{JvR16}.  The trajectory of the leading Fisher zero $a_1$
in figures \ref{figure-5} and \ref{figure3-5} follows an approximately
parabolic path with increasing length of the walk, and approaches the real axis at an
angle which was estimated to be about $55^o$ in the square lattice.  
Estimating this angle in the cubic lattice proved more difficult.  If the
estimate $a_c^+=1.306$ for the critical value of $a$ is used, then the
angle has a value of approximately $55^o$.  If this is not used, then
the angle appears to be larger, for example, approximately $64^o$ in figure
\ref{figure3-5}.

The cumulative angular distribution function of Fisher zeros is illustrated
in figure \ref{figure-6} for square lattice adsorbing walks, and
figure \ref{figure3-6} for cubic lattice adsorbing walks.  There is very 
little difference in these figures, and with increasing $n$ both cases
appear to approach the linear function $\sfrac{1}{2\pi}\,\phi + \sfrac{1}{2}$.
That is, the angular distribution is a constant function away from the 
angle $\phi=0$ in the limit as $n\to\infty$.  This is consistent with the 
graphs in figures \ref{figure-8} and \ref{figure3-8}, which becomes 
less noisy with increasing $n$.

The radial distributions of Fisher zeros were also plotted in figures \ref{figure-9}
and \ref{figure3-9}.  The fraction of zeros outside disks of given radius $r$
approaches a constant as $n$ increases, showing that the number density
of Fishes zeros decreases with increasing distance from the
the origin in the $a$-plane.

Results on the distribution of Yang-Lee zeros in the $a$-plane were 
similarly presented.  These zeros are distributed similarly to Fisher
zeros, and the results in theorems \ref{thm2.7} on the radial distribution
is similar to the radial distribution of Fisher zeros seen in theorem \Ref{thm2.6}.
The angular distribution of Yang-Lee zeros was also seen to be similar to
that of Fisher zeros, and the results on the distribution function are
given in theorem \ref{thm2.7aa}.
Numerical investigation on the Yang-Lee zeros produced results similar
to that seen for Fisher zeros.

\vspace{1cm}
\noindent{\bf Acknowledgements:} EJJvR acknowledges financial support 
from NSERC (Canada) in the form of a Discovery Grant.  

\vspace{1cm}

\noindent{\bf References}
\bibliographystyle{plain}
\bibliography{References}

\end{document}